\documentclass[journal]{IEEEtran}

\ifCLASSINFOpdf
\else
   \usepackage[dvips]{graphicx}
\fi
\usepackage{url}
\usepackage{graphicx} 
\usepackage{enumitem}
\usepackage{xcolor}
\usepackage{amsmath}
\usepackage{amsthm}
\usepackage{amsfonts}
\usepackage{amssymb}
\usepackage{multicol}
\usepackage{cite}
\usepackage{algorithm,algorithmic}
\usepackage{flushend}

\usepackage[caption=false]{subfig}

\newtheorem{theo}{Theorem}
\newtheorem{rem}{Remark}

\makeatletter
\newcommand*{\rom}[1]{\expandafter\@slowromancap\romannumeral #1@}
\makeatother

\begin{document}

\title{Distributed Filtering Design with Enhanced Resilience to Coordinated Byzantine Attacks}
\author{Ashkan Moradi, \IEEEmembership{Member, IEEE}, Vinay Chakravarthi Gogineni, \IEEEmembership{Member, IEEE}, Naveen K. D. Venkategowda, \IEEEmembership{Member, IEEE}, Stefan Werner, \IEEEmembership{Senior Member, IEEE}
\thanks{This work was supported by the Research Council of Norway. }
\thanks{Ashkan Moradi, Vinay Chakravarthi Gogineni, and Stefan Werner are with the Department of Electronic Systems,  Norwegian University of Science and Technology, Norway, E-mail: \{ashkan.moradi, vinay.gogineni, stefan.werner\}@ntnu.no. Stefan Werner is also with the Department of Signal Processing and Acoustics, Aalto University, Finland.} %
\thanks{Naveen K. D. Venkategowda is with the Department of Science and Technology, Link\"oping University, Sweden, E-mail: naveen.venkategowda@liu.se.}
}

\maketitle

\begin{abstract}
This paper proposes a Byzantine-resilient consensus-based distributed filter (BR-CDF) wherein network agents employ partial sharing of state parameters. We characterize the performance and convergence of the BR-CDF and study the impact of a coordinated data falsification attack. Our analysis shows that sharing merely a fraction of the states improves robustness against coordinated Byzantine attacks. In addition, we model the optimal attack strategy as an optimization problem where Byzantine agents design their attack covariance or the sequence of shared fractions to maximize the network-wide mean squared error (MSE). Numerical results demonstrate the accuracy of the proposed BR-CDF and its robustness against Byzantine attacks. Furthermore, the simulation results show that the influence of the covariance design is more pronounced when agents exchange larger portions of their states with neighbors. In contrast, the performance is more sensitive to the sequence of shared fractions when smaller portions are exchanged.
\end{abstract}


\IEEEpeerreviewmaketitle

\section{Introduction}

\IEEEPARstart{I}{n} recent years, the development of the internet of things (IoT) and machine learning techniques have led to the wide use of cyber-physical systems (CPS) in various infrastructures, such as smart grids, environment monitoring, and signal processing \cite{Humayed2017Cyber, Ding2020Secure, Farraj2017cyber, Hu2018State}. However, the high reliance of CPS on sensor cooperation makes them vulnerable to various security threats. As a result of malicious attacks, false information spreads throughout the network and threatens the integrity of the entire system~\cite{Rodriguez2015Physical}. Therefore, the study of security threats in CPS has gained considerable attention from academia and industry in the past few years \cite{Forti2017Distributed, Rahman2016Multi, Fawzi2014Secure, Moradi2022privacy, moradi2021distributed, moradi2021securingthe}. In CPSs, distributed filtering and secure estimation are becoming more prevalent due to their resilience to node failures and scalability~\cite{Liang2022Secure, Ding2019survey}.

Attack strategies influence the performance of CPSs and create complications in developing protection methods against malicious behavior~\cite{Zhao2017Resilient, Guan2017Distributed}. In CPSs, attacks can be divided into two groups: denial-of-service (DoS) attacks and integrity attacks. DoS attacks occur when the communication links between agents are blocked, and agents cannot exchange information~\cite{Chang2002DoS}. In contrast, integrity attacks occur when adversaries or malicious agents inject false information into the network~\cite{vempaty2013, moradi2023total}. A stealthy attack is also categorized as an integrity attack in which an adversary injects false information into a network without being detected.
Various studies have examined the impact of stealthy attacks on state estimation scenarios and investigated situations where adversaries design attacks to degrade the performance of the network~\cite{Mo2015performance, Deng2016False, Li2018False}.

An optimal attack design from an attacker perspective can aid in developing protection methods operating in worst-case scenarios. 
To this end, \cite{Li2018False} designs a false data-injection strategy that maximizes the trace of error covariance in a remote state estimation scenario. 
An event-based stealthy attack strategy that degrades estimation accuracy is proposed in~\cite{Cheng2019stealthy}, while~\cite{Srikantha2019novel} develops a stealthy attack strategy and its optimal defense mechanism by exploiting power grid vulnerabilities. Furthermore,  \cite{Guo2016Optimal}  proposes a linear stealthy attack strategy and its feasibility constraints. Moreover, \cite{Hu2018State} proposes an optimal attack strategy and sufficient conditions to limit the resulting estimation error. In \cite{Ni2019performance}, authors investigate the impact of reset attacks on cyber-physical systems, while \cite{Moradi2019Coordinated} designs a stealthy attack that maximizes the network-wide estimation error by jointly selecting the optimal subset of Byzantine agents and perturbation covariance matrices. In \cite{Choraria2022Design},  a false data-injection attack on a distributed CPS is proposed that enforces the local state estimates to remain within a pre-specified range.

It is also vital to analyze how countermeasures taken by agents can reduce the impact of the attack. One approach is to detect adversaries and then implement corrective measures~\cite{Kurt2018quickest,Kurt2018hybrid, Aktukmak2021Sequential}. In~\cite{Chen2018Resilient}, for example, attack detection is achieved through trusted agents that raise a flag when an adversary is detected. Alternatively, agents can safeguard information and guarantee system performance by running attack-resilient algorithms~\cite{Li2016Stochastic,Li2018Detection,Yang2019filtering,Chen2018Resilient,Chen2019Resilient,Barboni2020Detection,Shang2021Optimal}.  
In~\cite{Li2016Stochastic,Li2018Detection}, attack-resilient remote state estimators are studied, where the former proposes a stochastic detector with a random threshold to determine whether to fuse the received data, and the latter detects malicious agents using the statistical correlation between trusted agents. Using a probabilistic protector for each sensor, \cite{Yang2019filtering}  proposes a robust distributed state estimator that decides whether to use data from neighbors based on their innovation signals.
Moreover, a Byzantine-resilient distributed state estimation algorithm is proposed in~\cite{Shahrampour2020Finite} that employs an internal iteration loop within the local aggregation process to compute the trimmed-mean among neighbors.

Providing an extra procedure to protect the system from adversaries~\cite{Kurt2018quickest,Kurt2018hybrid, Aktukmak2021Sequential, Chen2018Resilient} can increase the computational load of the agent and make the algorithm undesirable for resource-constrained scenarios. 
To resolve this issue, works in~\cite{Chen2019Resilient, chen2019resilient1} reduce the weight assigned to measurements whose norm exceeds a certain threshold, limiting the impact of falsified observations,  to provide resilience against measurement attacks. Furthermore, the secure state estimation problem in~\cite{lee2020fully} is solved by a local observer that achieves robustness against malicious agents by employing the median of its local estimates. Our primary motivation for conducting this study was to develop an algorithm that provides robustness against malicious adversaries without imposing an extra computation burden on agents.

Even though multi-agent distributed systems are robust to dynamic changes in the network, they are reliant on local interactions that consume power and bandwidth~\cite{Feng2012survey}. Partial-sharing-based approaches, originally proposed in~\cite{PdLMS, PdRLS},  reduce local communication overhead by sharing only a fraction of information during each inter-agent interaction. The simplicity of implementation and efficiency of computation make partial-sharing strategies prevalent in distributed processing scenarios \cite{gogineni2022communication, Vinay_IoT}. To the best of our knowledge, partial-sharing-based approaches have not been investigated in an adversarial environment. Additionally, the lack of computationally light distributed algorithms that are robust to coordinated attacks inspired us to conduct this research.

This paper proposes a Byzantine-resilient consensus-based distributed filter (BR-CDF) where agents exchange a fraction of their state estimates at each instant. We study the convergence of the BR-CDF and characterize its performance under a coordinated data falsification attack. 
In addition, we design the optimal attack by solving an optimization problem where Byzantine agents cooperate on designing the covariance of their falsification data or the order of the information fractions they share. The Byzantine agent is a legitimate network agent that shares false information with neighbors to degrade the overall performance of the network. The numerical results validate the theoretical findings and illustrate the robustness of the  BR-CDF algorithm.

The remainder of the article is organized as follows. Section~II investigates the system model and the attack strategy, while Section~III proposes the Byzantine-resilient consensus-based distributed filter. Section~IV analyzes the stability and performance of the BR-CDF algorithm and derives its convergence conditions. The performance of the  BR-CDF algorithm is investigated under data falsification attacks in Section~V, and Section~VI develops an optimal coordinated attack strategy. Simulation results are presented in Section VII, and Section VIII concludes the article.

\noindent\textit{\textbf{Mathematical Notations}}: Scalars are denoted by lowercase letters, column vectors by bold lowercase, and matrices by bold uppercase.  Transpose and inverse operators are denoted by $(\cdot)^\text{T}$ and $(\cdot)^{-1}$, respectively.  The trace operator is denoted by $\text{tr}(\cdot)$, whereas $\otimes$ indicates the Kronecker product and $\odot$ is the Hadamard product. The $ij$th element of the matrix $\mathbf{A}$ is  denoted by $[\mathbf{A}]_{ij}$. The symbol $\boldsymbol{1}_{L}$ represents the $L\times 1$ column vector with all entries equal to one and $\mathbf{I}_L$ is the $L \times L$ identity matrix.  Matrices $\textsf{diag}(\mathbf{a})$ and $\textsf{diag}(\{\mathbf{A}_i\}_{i=1}^L)$ denote diagonal and block-diagonal matrices whose respective diagonals are the elements of vector $\mathbf{a}$ and matrices $\mathbf{A}_1, \mathbf{A}_2, \hdots, \mathbf{A}_L$. The set of all real numbers is denoted as $\mathbb{R}$ and $\mathbb{E}\{\cdot\}$ is the statistical expectation operator. A positive semidefinite matrix $\mathbf{A}$ is denoted by $\mathbf{A} \succcurlyeq 0$ and $\mathbf{A} \succcurlyeq \mathbf{B}$ indicates that $\mathbf{A}-\mathbf{B}$ is a positive semidefinite matrix. The inequality $\mathbf{A} \leq \mathbf{B}$ denotes an element-wise inequality for corresponding elements in matrices $\mathbf{A}$ and $\mathbf{B}$. The maximum and minimum eigenvalues of square matrix $\mathbf{A}$ are denoted by $\lambda_{\max}(\mathbf{A})$ and $\lambda_{\min}(\mathbf{A})$, respectively. Vector $\text{vec}(\mathbf{A})$ denotes a column vector consisting of elements of the matrix $\mathbf{A}$ and $\text{vec}^{-1}(\cdot)$ denotes  the inverse of $\text{vec}(\cdot)$ operator.

\section{System Model and Byzantine Attack Strategy}
\label{Sys_Model_DKF}
Consider a multi-agent network consisting of $L$ agents attempting to estimate a dynamic system state through their local observations. The network is modeled as an undirected graph $\mathcal{G}(\mathcal{V}, \mathcal{E})$, where $\mathcal{V}$ is the set of all agents and pairs in set $\mathcal{E}$ represent communication links between agents. The network adjacency matrix is denoted by $\mathbf{E}$ and $\mathbf{D}$ is a diagonal matrix whose diagonal entries are the degrees of corresponding agents. The neighbor set ${\cal N}_i$ contains agents connected to agent $i$ within a single hop, excluding the agent itself. 

The state-space model, characterizing the dynamics of the state vector and observation sequences at each agent $i$ and time instant $k$, is given by
\begin{equation}\label{Eq1}
	\begin{aligned}
	\mathbf{x}(k + 1) &= \mathbf{A} \mathbf{x}(k) +\mathbf{w}(k) \\
	\mathbf{y}_i(k) &= \mathbf{H}_i\mathbf{x}(k) + \mathbf{v}_i(k)
	\end{aligned}
\end{equation}
where $\mathbf{x}(k) \in \mathbb{R}^m$ is the state, $\mathbf{y}_i(k) \in \mathbb{R}^n$ is the local observation, $\mathbf{A} \in \mathbb{R}^{m \times m}$  is the state  matrix, and $\mathbf{H}_i\in \mathbb{R}^{n \times m}$ is the observation matrix. The state noise $\mathbf{w}(k)$ and observation noise  $\mathbf{v}_i(k)$ are mutually independent zero-mean Gaussian processes with covariance matrices $\mathbf{Q} \in \mathbb{R}^{m \times m}$ and $\mathbf{R}_i \in \mathbb{R}^{n \times n}$, respectively. Network agents can employ a consensus-based distributed Kalman filter (CDF) to estimate $\mathbf{x}(k)$ in a collaborative manner~\cite{olfati2009kalman}. Accordingly, the state estimate at agent $i$ is given by
\begin{equation}
\label{f2}
	\begin{aligned}
	\mathbf{\hat{x}}_i(k+1)=\mathbf{A}\mathbf{\hat{x}}_i&(k)+\mathbf{K}_i(k)\big( \mathbf{y}_i(k)-\mathbf{H}_i\mathbf{\hat{x}}_i(k)\big) \\ 
	&+ \mathbf{C}_i\textstyle \sum_{j \in \mathcal{N}_i}\big(\mathbf{\bar{x}}_j(k)-\mathbf{\hat{x}}_i(k)\big)
	\end{aligned}
\end{equation}
where  $\mathbf{C}_i \in \mathbb{R}^{m \times m}$ denotes the consensus gain, $\mathbf{K}_i(k)\in \mathbb{R}^{m \times n}$ is the Kalman gain, and  $\bar{\mathbf{x}}_j(k)$ is the received state estimates from neighboring agent $j$.

To analyze the impact of data falsification attacks on network performance, the attack model needs to be specified. For this purpose, we assume a subset of agents to be Byzantines, i.e., ${\cal B}\subseteq \mathcal{V}$, that intend to disrupt the performance of the entire network~\cite{vempaty2013}. Fig.~\ref{Net_topo} shows the dynamic of the information exchange in a network with Byzantine agents. As seen in~Fig.~\ref{Net_topo}, a regular agent shares the actual value of the state estimate  $\mathbf{\hat{x}}_j(k)$ with its neighbors. In contrast, a Byzantine agent $j \in {\cal B}$ shares a perturbed version of its state estimate with neighbors; in particular, the shared information at each agent $j$ is denoted by  
\begin{equation}\label{shared_inf}
	\mathbf{\bar{x}}_j(k) = \begin{cases}
	\mathbf{\hat{x}}_j(k)+\boldsymbol{\delta}_j(k)  & j \in \mathcal{B} \\
	\mathbf{\hat{x}}_j(k)  &  j \notin \mathcal{B} 
	\end{cases} \quad
\end{equation}
where $\boldsymbol{\delta}_j(k)$  denotes the perturbation sequence. To maximize the attack stealthiness, i.e., the ability to evade detection, the perturbation sequence is drawn from a zero-mean Gaussian distribution with covariance $\mathbf{\Sigma}_j=\mathbb{E}\{\boldsymbol{\delta}_j(k)\boldsymbol{\delta}_j^{\text{T}}(k)\}$~\cite{Chen2018a, Bai2017}. Moreover, to further degrade the network performance, Byzantines can cooperate in designing the attack strategy. The network-wide coordinated attack covariance is denoted by $\boldsymbol{\Sigma}= \mathbb{E}\{\boldsymbol{\delta}(k)\boldsymbol{\delta}^\text{T}(k)\}$ where  
$\boldsymbol{\delta}(k)=[\boldsymbol{\delta}_1^\text{T}(k),\cdots,\boldsymbol{\delta}_L^\text{T}(k)]^\text{T}$
is the network-wide attack sequence with  $\boldsymbol{\delta}_j(k)=\mathbf{0}$ if $j \notin \mathcal{B}$.

\begin{figure}[!t]
	\centering
	\includegraphics[width=.4\textwidth]{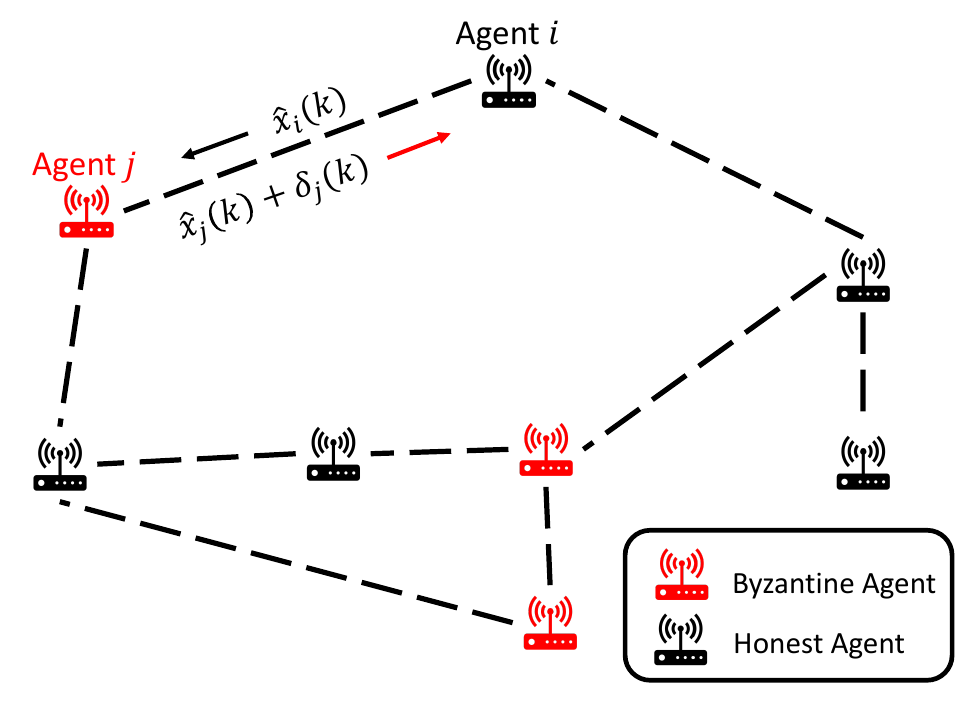}
	\caption{Network topology in the presence of Byzantine agents.}
	\label{Net_topo}
\end{figure}

\section{Byzantine-Resilient consensus-based distributed Kalman filter }

By applying partial sharing of information to state estimates in~\eqref{f2}, we reduce the information flow between agents at a given instant while maintaining the advantages of cooperation~\cite{PdLMS, PdRLS}. In particular, each agent only shares a fraction of its state estimate with neighbors rather than the entire vector (i.e., $l$ entries of $\mathbf{\hat{x}}_j(k)\in \mathbb{R}^m$, with $l \leq m$). Although partial-sharing was originally introduced to reduce inter-agent communication overhead, we show that adopting this idea in the current setting improves robustness to Byzantine attacks.

The state estimate entry selection process is performed at each agent $j$ using a selection matrix $\mathbf{S}_j(k)$ of size $m \times m$, whose main diagonal contains $l$ ones and $m-l$ zeros. The ones on the main diagonal of $\mathbf{S}_j(k)$ specify the entries of the state estimate $\mathbf{\hat{x}}_j(k)$ to be shared with neighbors. Selecting $l$ entries from $m$ can either be done stochastically or sequentially, as in \cite{PdLMS} and  \cite{PdRLS}. In this paper, we use uncoordinated partial-sharing, which is a special case of stochastic partial-sharing \cite{PdLMS}. In uncoordinated partial-sharing, each agent $j$ is initialized with random selection matrices. The selection matrix at the current time instant, i.e., $\mathbf{S}_j(k)$, can be obtained by performing $\tau$ right-circular shift operations on the main diagonal of the selection matrix used in the previous time instant. In other words, if $\mathbf{s}_j(k) \in \mathbb{R}^m$ contains the main diagonal elements of $\mathbf{S}_j(k)$ at the current instant, then $\mathbf{s}_j(k)=\text{right-circular shift} \{\mathbf{s}_j(k-1), \tau\}$. Then the selection matrix at the current instant can be constructed as $\mathbf{S}_j(k)=\textsf{diag}\{\mathbf{s}_j(k)\}$. This allows each agent $j$ to share only the initial selection matrix  $\mathbf{S}_j(k)$ with its neighbors, and maintain a record of the indices of parameters shared without needing any additional mechanisms. As a result, the frequency of each entry of the state estimate being shared is equal to $p_e=\frac{l}{m}$.

Due to partial sharing, every agent receives a fraction of the perturbed state estimate vectors from its neighbors, i.e.,  $\mathbf{S}_j(k)\mathbf{\bar{x}}_j(k)$. Thus, unlike~\eqref{shared_inf}, the received information here must be compensated to fill in the missing elements. At each agent $i$, the missing values from neighbors $(\mathbf{I}-\mathbf{S}_j(k))\mathbf{\hat{x}}_j(k)$ are replaced by  $(\mathbf{I}-\mathbf{S}_j(k))\mathbf{\hat{x}}_i(k)$. 
Subsequently, the state update at agent $i$, as in~\eqref{f2},  is modified as follows.
\begin{align}
    \mathbf{\hat{x}}_i(k+1)=&\mathbf{A}\mathbf{\hat{x}}_i(k)+\mathbf{K}_i(k)\big( \mathbf{y}_i(k)-\mathbf{H}_i\mathbf{\hat{x}}_i(k)\big) \nonumber\\ 
	&+\mathbf{C}_i\textstyle \sum_{j \in \mathcal{N}_i}\mathbf{S}_j(k)\big(\mathbf{\hat{x}}_j(k)-\mathbf{\hat{x}}_i(k)\big)\label{f5}\\
	&+\mathbf{C}_i\textstyle \sum_{j \in \mathcal{N}_i}\mathbf{S}_j(k)\boldsymbol{\delta}_j(k) \nonumber
\end{align}
At each agent $i$, the Kalman gain  is obtained by minimizing the trace of the estimation error covariance $\mathbf{P}_i(k) \triangleq \mathbb{E}\{\mathbf{e}_i(k)\mathbf{e}^\text{T}_i(k)\}$ with the estimation error evolving as 
\begin{align}
\label{Error_actual}
\mathbf{e}_i(k+1) \triangleq &  \mathbf{\hat{x}}_i(k+1) -\mathbf{x}(k+1) \nonumber\\ =& \mathbf{F}_i(k)\mathbf{e}_i(k)+ \mathbf{C}_i\textstyle \sum_{j \in \mathcal{N}_i}\mathbf{S}_j(k)\mathbf{e}_j(k)\\
&+\mathbf{K}_i(k)\mathbf{v}_i(k)-\mathbf{w}(k)+ \mathbf{C}_i\textstyle \sum_{j \in \mathcal{N}_i}\mathbf{S}_j(k)\boldsymbol{\delta}_j(k) \nonumber
\end{align}
where 
\begin{equation}\label{F_eq}
\mathbf{F}_i(k)=\mathbf{A}-\mathbf{K}_i(k)\mathbf{H}_i-\mathbf{C}_i\textstyle \sum_{j \in \mathcal{N}_i}\mathbf{S}_j(k)
\end{equation}
Accordingly, using \eqref{Error_actual}, we can obtain the evolution of the estimation error covariance at agent $i$ as follows, 
\begin{align}
\mathbf{P}_i&(k+1)	=\mathbf{F}_i(k)\mathbf{P}_i(k)\mathbf{F}_i^\text{T}(k) +\mathbf{K}_i(k)\mathbf{R}_i\mathbf{K}_i^\text{T}(k)+\mathbf{Q} \label{f6}\\
&+\Delta \mathbf{P}_i(k)+ \mathbf{C}_i \textstyle  \sum_{s \in \mathcal{N}_i}\textstyle  \sum_{p \in \mathcal{N}_i}\mathbf{S}_s(k) \boldsymbol{\Sigma}_{sp} \mathbf{S}_p^\text{T}(k)\mathbf{C}_i^\text{T} \nonumber
\end{align}
where $\boldsymbol{\Sigma}_{sp}=\mathbb{E}\{\boldsymbol{\delta}_s(k)\boldsymbol{\delta}_p^\text{T}(k)\}$ and
\begin{align}
\Delta \mathbf{P}_i(k)	=&  \mathbf{F}_i(k)\textstyle \sum_{j \in \mathcal{N}_i} \mathbf{P}_{ij}(k)\mathbf{S}_j^\text{T}(k)\mathbf{C}_i^\text{T}\nonumber \\
&+\mathbf{C}_i\textstyle  \sum_{j \in \mathcal{N}_i}\mathbf{S}_j(k)\mathbf{P}_{ji}(k)\mathbf{F}_i^\text{T}(k)\nonumber\\
&+\mathbf{C}_i \textstyle  \sum_{s \in \mathcal{N}_i}\textstyle  \sum_{p \in \mathcal{N}_i}\mathbf{S}_s(k)\mathbf{P}_{sp}(k)\mathbf{S}_p^\text{T}(k)\mathbf{C}_i^\text{T} \nonumber
\end{align}
Similarly, the cross-terms of the error covariance, i.e., $\mathbf{P}_{ij}(k) \triangleq \mathbb{E}\{\mathbf{e}_i(k)\mathbf{e}^\text{T}_j(k)\}$,  evolve as 
\begin{align}
\mathbf{P}_{ij}(k+1)	=&\mathbf{F}_i(k)\mathbf{P}_{ij}(k)\mathbf{F}_j^\text{T}(k) +\mathbf{Q} +\Delta\mathbf{P}_{ij}(k) \nonumber \\
&+ \mathbf{C}_i \textstyle  \sum_{s \in \mathcal{N}_i}\textstyle  \sum_{p \in \mathcal{N}_j}\mathbf{S}_s(k) \boldsymbol{\Sigma}_{sp} \mathbf{S}_p^\text{T}(k)\mathbf{C}_j^\text{T}\nonumber
\end{align}
with 
\begin{equation*}
\begin{aligned}
\Delta\mathbf{P}_{ij}(k)=& \mathbf{F}_i(k)\textstyle \sum_{p \in \mathcal{N}_j} \mathbf{P}_{ip}(k)\mathbf{S}_p^\text{T}(k)\mathbf{C}_j^\text{T} \\
&+\mathbf{C}_i\textstyle  \sum_{s \in \mathcal{N}_i}\mathbf{S}_s(k)\mathbf{P}_{sj}(k)\mathbf{F}_j^\text{T}(k)\\
&+\mathbf{C}_i \textstyle  \sum_{s \in \mathcal{N}_i}\textstyle  \sum_{p \in \mathcal{N}_j}\mathbf{S}_s(k)\mathbf{P}_{sp}(k)\mathbf{S}_p^\text{T}(k)\mathbf{C}_j^\text{T}
\end{aligned}
\end{equation*}
Differentiating the trace of \eqref{f6} with respect to $\mathbf{K}_i(k)$ gives
\begin{align}
\mathbf{K}_i^*(k)&=\bigg( \left(\mathbf{A}-\mathbf{C}_i\textstyle \sum_{j \in \mathcal{N}_i}\mathbf{S}_j(k)\right)\mathbf{P}_i(k)\nonumber \\
&\quad \quad \quad+ \mathbf{C}_i\textstyle  \sum_{j \in \mathcal{N}_i}\mathbf{S}_j(k)\mathbf{P}_{ji}(k)\bigg)\mathbf{H}_i^\text{T} \mathbf{M}_i^{-1}(k) \nonumber
\end{align}
with
$\mathbf{M}_i(k)=\mathbf{R}_i+\mathbf{H}_i\mathbf{P}_i(k)\mathbf{H}_i^\text{T}$.

The distributed Kalman filter based on partial sharing is summarized by \eqref{f5}--\eqref{f6}. We see that the local covariance update in~\eqref{f6} requires access to cross-term covariance matrices of neighbors, resulting in considerable communication overhead.
To reduce the communication overhead, for sufficiently small gain values, i.e., $\mathbf{C}_i$, we can ignore the term  $\Delta\mathbf{P}_i(k)$ in~\eqref{f6} and the last term of $\mathbf{F}_i(k)$ in~\eqref{F_eq}~\cite{olfati2009kalman}, i.e., we have 
\begin{align}\label{Pi_BRDKF}
\mathbf{P}_i(k+1)=&\mathbf{\hat{F}}_i(k) \mathbf{P}_i(k)\mathbf{\hat{F}}_i^\text{T}(k)+\mathbf{K}_i(k)\mathbf{R}_i\mathbf{K}_i^\text{T}(k) +\mathbf{Q} \\
&+\mathbf{C}_i \textstyle  \sum_{s \in \mathcal{N}_i}\textstyle  \sum_{p \in \mathcal{N}_i}\mathbf{S}_s(k) \boldsymbol{\Sigma}_{sp} \mathbf{S}_p^\text{T}(k)\mathbf{C}_i^\text{T}\nonumber 
\end{align}
with $\mathbf{\hat{F}}_i(k)=\mathbf{A}-\mathbf{K}_i(k)\mathbf{H}_i$. Accordingly, the optimal Kalman gain reduces to
\begin{equation}
\label{Kal_gain}
\mathbf{K}_i(k)= \mathbf{A}\mathbf{P}_i(k)\mathbf{H}_i^\text{T} \left(\mathbf{R}_i+\mathbf{H}_i\mathbf{P}_i(k)\mathbf{H}_i^\text{T}\right)^{-1}
\end{equation}
 With the above approximations, we obtain a distributed consensus-based Kalman filter, albeit suboptimal~\cite{olfati2009kalman, olfati2007distributed}, that only requires local variables in the error covariance update at each agent. It is worth noting that the last term in~\eqref{Pi_BRDKF} is only used to characterize the impact of the perturbation covariances, and since the attack is stealthy from the perspective of an agent, it is excluded from the filtering algorithm.
 As a result, in addition to the initial selection matrix $\mathbf{S}_j(0)$, each agent $j$  shares a fraction of the perturbed state estimate, i.e.,  $\mathbf{S}_j(k)\mathbf{\bar{x}}_j(k)$, with its neighbors at each instant.
The proposed  BR-CDF algorithm with reduced communication is summarized in Algorithm~\ref{Alg1}.  We shall see that the  BR-CDF in Algorithm~\ref{Alg1} performs closely to the solution that shares all necessary variables.
 
\begin{algorithm}[t]
\caption{ BR-CDF Algorithm}
\label{Alg1}
\begin{algorithmic}
\renewcommand{\algorithmicrequire}{ For}
\REQUIRE each agent $i \in \mathcal{N}$ 
\renewcommand{\algorithmicrequire}{\textbf{Initialize:}}
\REQUIRE   $\mathbf{\hat{x}}_{i}(0)=\mathbf{x}_0$, $\mathbf{P}_{i}(0)=\mathbf{P}_0$,  $\mathbf{S}_{i}(0)=\textsf{diag}(\mathbf{s}_i(0))$, and share $\mathbf{S}_i(0)$ with $j \in \mathcal{N}_i$
\FORALL{$k > 0$}
\STATE For all $j \in \mathcal{N}_i$ receive  $ \mathbf{S}_j(k)\mathbf{\bar{x}}_j(k)$ 
\STATE 
$\mathbf{K}_i(k)= \mathbf{A}\mathbf{P}_i(k)\mathbf{H}_i^\text{T} \left(\mathbf{R}_i+\mathbf{H}_i\mathbf{P}_i(k)\mathbf{H}_i^\text{T}\right)^{-1}$
\STATE \normalsize{Update the state estimate}
\begin{align}
\hspace{-2mm}\mathbf{\hat{x}}_i(k+1)=&\mathbf{A}\mathbf{\hat{x}}_i(k)+\mathbf{K}_i(k)\big( \mathbf{y}_i(k)-\mathbf{H}_i\mathbf{\hat{x}}_i(k)\big) \label{state_upd}\\ 
	&+\mathbf{C}_i\textstyle \sum_{j \in \mathcal{N}_i}\big(\mathbf{S}_j(k)\mathbf{\bar{x}}_j(k)-\mathbf{S}_j(k)\mathbf{\hat{x}}_i(k)\big)\nonumber
\end{align}
\STATE $\mathbf{\hat{F}}_i(k)=\mathbf{A}-\mathbf{K}_i(k)\mathbf{H}_i$
\STATE Update the local covariance
\begin{equation}\label{Pi_algo}
\mathbf{P}_i(k+1)=\mathbf{\hat{F}}_i(k) \mathbf{P}_i(k)\mathbf{\hat{F}}_i^\text{T}(k)+\mathbf{K}_i(k)\mathbf{R}_i\mathbf{K}_i^\text{T}(k) +\mathbf{Q} 
\end{equation}
\STATE $\mathbf{s}_i(k+1)=\text{right-circular shift} \{\mathbf{s}_i(k), \tau\}$
\STATE $\mathbf{S}_i(k+1)=\textsf{diag}\left(\mathbf{s}_i(k+1)\right)$
\STATE Share $ \mathbf{S}_{i}(k+1)\mathbf{\bar{x}}_i(k+1)$ with all $j \in \mathcal{N}_i$
\ENDFOR
\end{algorithmic} 
\end{algorithm}

\section{Stability and Performance Analysis}
This section provides a detailed stability analysis of the  BR-CDF algorithm. For this purpose, we make the following assumption:\newline
\textbf{Assumption~1}: 
The selection matrix $\mathbf{S}_{i}(k)$ for all $i\in\mathcal{N}$ is independent of any other data, and the selection matrices $\mathbf{S}_{i}(k)$ and $\mathbf{S}_{j}(s)$ for all $i \neq j$ and $k \neq s$ are independent.

\noindent Our main result on the stability of the proposed  BR-CDF algorithm is summarized by the following theorem.

\begin{theo}
\label{Th1}
Consider the BR-CDF in Algorithm~\ref{Alg1} with consensus gain $\mathbf{C}_i=\gamma \mathbf{A}\mathbf{\bar{M}}_i^{-1}(k)$, where $\mathbf{\bar{M}}_i(k)=\mathbf{P}_i^{-1}(k)+\mathbf{H}_i^{\text{T}}\mathbf{R}_i^{-1}\mathbf{H}_i$. Then, for a sufficiently small $\gamma$, the error dynamics of
the BR-CDF is globally asymptotically stable and all local estimators
asymptotically reach a consensus on state estimates, i.e., $\mathbf{\hat{x}}_1(k)=\mathbf{\hat{x}}_2(k)=\cdots=\mathbf{\hat{x}}_L(k)=\mathbf{x}(k)$.
\end{theo}
\begin{proof}
The proof begins by analyzing the dynamics of the estimation error in the absence of noise~\cite{olfati2009kalman}. Given the consensus-based Kalman approach in~\eqref{state_upd}, the  estimation error dynamics, without noise, at each agent $i$,  can be written~as
\begin{align}
\mathbf{\hat{e}}_i(k+1)&=\mathbf{\hat{F}}_i(k)\mathbf{\hat{e}}_i(k)+\mathbf{C}_i \textstyle\sum_{j \in \mathcal{N}_i}\mathbf{S}_j(k)\big(\mathbf{\hat{e}}_j(k)-\mathbf{\hat{e}}_i(k)\big)\nonumber \\
&=\mathbf{\hat{F}}_i(k)\mathbf{\hat{e}}_i(k)+\mathbf{C}_i \mathbf{u}_i(k)\label{err_dyna}
\end{align}
where $\mathbf{u}_i(k)=\textstyle \sum_{j \in \mathcal{N}_i}\mathbf{S}_j(k)\big(\mathbf{\hat{e}}_j(k)-\mathbf{\hat{e}}_i(k)\big)$. 
Our goal is to determine $\mathbf{C}_i$ such that the estimation error dynamic in~\eqref{err_dyna} is stable. 
Following the  approach in~\cite{olfati2009kalman}, we use 
\begin{equation}
   \mathbf{V}(\mathbf{\hat{e}}(k))= \textstyle\sum_{i=1}^{L}\mathbf{\hat{e}}_i^{\text{T}}(k)\mathbf{P}_i^{-1}(k)\mathbf{\hat{e}}_i(k)
\end{equation}
as a candidate Lyapunov function for \eqref{err_dyna} where the network-wide stacked error is $\mathbf{\hat{e}}(k)\triangleq[\mathbf{\hat{e}}_1^\text{T}(k), \cdots,  \mathbf{\hat{e}}_L^\text{T}(k) ]^\text{T}$. 
We can then express $\delta \mathbf{V}(\mathbf{\hat{e}}(k))\triangleq\mathbf{V}(\mathbf{\hat{e}}(k+1))-\mathbf{V}(\mathbf{\hat{e}}(k))$ as  
\begin{align}
    \delta \mathbf{V}(\mathbf{\hat{e}}(k))=&\textstyle\sum_{i=1}^{L}\bigg(\mathbf{\hat{e}}_i^{\text{T}}(k+1)\mathbf{P}_i^{-1}(k+1)\mathbf{\hat{e}}_i(k+1)\nonumber\\
    &\quad \quad \quad -\mathbf{\hat{e}}_i^{\text{T}}(k)\mathbf{P}_i^{-1}(k)\mathbf{\hat{e}}_i(k)\bigg)\label{delta_Lyapanov}
\end{align} 
By substituting \eqref{Kal_gain} into \eqref{Pi_algo} 
and by employing the matrix inversion lemma, we have 
\begin{equation}\label{former_Lemma_1}
\mathbf{P}_i(k+1)=\mathbf{A}\mathbf{\bar{M}}_i^{-1}(k)\mathbf{A}^{\text{T}}+\mathbf{Q}
\end{equation}
where $\mathbf{\bar{M}}_i(k)=\mathbf{P}_i^{-1}(k)+\mathbf{H}_i^{\text{T}}\mathbf{R}_i^{-1}\mathbf{H}_i$. Subsequently, replacing \eqref{former_Lemma_1}, without noise, and \eqref{err_dyna} into \eqref{delta_Lyapanov} yields
\begin{align}    
    &\delta \mathbf{V}(\mathbf{\hat{e}}(k))=\textstyle\sum_{i=1}^{L}\hspace{-1mm}\bigg( \big(\mathbf{\hat{F}}_i(k)\mathbf{\hat{e}}_i(k)+\mathbf{C}_i\mathbf{u}_i(k)\big)^{\text{T}} \big(\mathbf{A}\mathbf{\bar{M}}_i^{-1}(k)\mathbf{A}^{\text{T}}\big)^{-1} \nonumber\\
    &\quad \quad  \big(\mathbf{\hat{F}}_i(k)\mathbf{\hat{e}}_i(k)+\mathbf{C}_i\mathbf{u}_i(k)\big) -\mathbf{\hat{e}}_i^{\text{T}}(k)\mathbf{P}_i^{-1}(k)\mathbf{\hat{e}}_i(k)\bigg)\label{delta_V16}
\end{align}
Furthermore, by substituting \eqref{Kal_gain} into $\mathbf{\hat{F}}_i(k)=\mathbf{A}-\mathbf{K}_i(k)\mathbf{H}_i$ 
and employing the matrix inversion lemma, we obtain 
$\mathbf{\hat{F}}_i(k)=\mathbf{A}\mathbf{\bar{M}}_i^{-1}(k)\mathbf{P}_i^{-1}(k)$.
Consequently, by replacing $\mathbf{\hat{F}}_i(k)$ into \eqref{delta_V16} and after some algebraic manipulations, we obtain
\begin{align}
    \delta\mathbf{V}&(\mathbf{\hat{e}}(k))=\textstyle\sum_{i=1}^{L}\bigg(\mathbf{\hat{e}}_i^{\text{T}}(k)\mathbf{P}_i^{-1}(k)\mathbf{\bar{M}}_i^{-1}(k)\mathbf{P}_i^{-1}(k)\mathbf{\hat{e}}_i(k)\nonumber\\
    &-\mathbf{\hat{e}}_i^{\text{T}}(k)\mathbf{P}_i^{-1}(k)\mathbf{\hat{e}}_i(k)+\mathbf{\hat{e}}_i^{\text{T}}(k)\mathbf{P}_i^{-1}(k)\mathbf{A}^{-1}\mathbf{C}_i\mathbf{u}_i(k) \nonumber\\
    &  + \mathbf{u}_i^{\text{T}}(k)\mathbf{C}_i^{\text{T}}\mathbf{A}^{-\text{T}}\mathbf{P}_i^{-1}(k)\mathbf{\hat{e}}_i(k) \nonumber \\
    &  +\mathbf{u}_i^{\text{T}}(k)\mathbf{C}_i^{\text{T}}\mathbf{A}^{-\text{T}}\mathbf{\bar{M}}_i(k)\mathbf{A}^{-1}\mathbf{C}_i\mathbf{u}_i(k)\bigg)\nonumber \\
    =&\textstyle\sum_{i=1}^{L}\bigg(-\mathbf{\hat{e}}_i^{\text{T}}(k)\big(\mathbf{P}_i(k)+(\mathbf{H}_i^{\text{T}}\mathbf{R}_i^{-1}\mathbf{H}_i)^{-1} \big)^{-1}\mathbf{\hat{e}}_i(k) \nonumber \\
    & \quad\quad\quad+2\mathbf{\hat{e}}_i^{\text{T}}(k)\mathbf{P}_i^{-1}(k)\mathbf{A}^{-1}\mathbf{C}_i\mathbf{u}_i(k) \label{F15} \\
    &\quad\quad\quad+\mathbf{u}_i^{\text{T}}(k)\mathbf{C}_i^{\text{T}}\mathbf{A}^{-\text{T}}\mathbf{\bar{M}}_i(k)\mathbf{A}^{-1}\mathbf{C}_i\mathbf{u}_i(k)\bigg)\nonumber
\end{align}
With an appropriate choice of consensus gain, i.e., $\mathbf{C}_i=\gamma \mathbf{A}\mathbf{\bar{M}}_i^{-1}(k)$, and a proper selection of $\gamma>0$, all terms of~\eqref{F15} become negative semidefinite.
Subsequently, we have 
\begin{align}
    \delta\mathbf{V}(\mathbf{\hat{e}}&(k))= \textstyle\sum_{i=1}^{L}\bigg(-\mathbf{\hat{e}}_i^{\text{T}}(k)\big(\mathbf{P}_i(k)+(\mathbf{H}_i^{\text{T}}\mathbf{R}_i^{-1}\mathbf{H}_i)^{-1} \big)^{-1}\mathbf{\hat{e}}_i(k) \nonumber \\
    &\quad \quad+2\gamma\mathbf{\hat{e}}_i^{\text{T}}(k)\big( \mathbf{I}+\mathbf{H}_i^{\text{T}}\mathbf{R}_i^{-1}\mathbf{H}_i\mathbf{P}_i(k)\big)^{-1}\mathbf{u}_i(k) \label{f14} \\
    &\quad \quad+\gamma^2 \mathbf{u}_i^{\text{T}}(k)\big(\mathbf{P}_i^{-1}(k)+\mathbf{H}_i^{\text{T}}\mathbf{R}_i^{-1}\mathbf{H}_i\big)^{-1}\mathbf{u}_i(k) \bigg)\nonumber
\end{align}
By defining $\mathbf{S}(k)\triangleq \textsf{diag}(\{ \mathbf{S}_i(k)\}_{i=1}^L)$, \eqref{f14} becomes
\begin{align}
    \delta &\mathbf{V}(\mathbf{\hat{e}}(k))=-2\gamma \mathbf{\hat{e}}^{\text{T}}(k)\boldsymbol{\Lambda}_{\text{\rom{3}}}(k)(\mathbf{L}\otimes\mathbf{I})\mathbf{S}(k)\mathbf{\hat{e}}(k)\label{Lya_vari} \\
    &-\mathbf{\hat{e}}^{\text{T}}(k) \bigg( \boldsymbol{\Lambda}_{\text{\rom{1}}}(k)-\gamma^2\mathbf{S}(k)(\mathbf{L}\otimes\mathbf{I})\boldsymbol{\Lambda}_{\text{\rom{2}}}(k)(\mathbf{L}\otimes\mathbf{I})\mathbf{S}(k)\bigg)\mathbf{\hat{e}}(k) \nonumber
\end{align}
where $\mathbf{L}=\mathbf{D}-\mathbf{E}$ is the network Laplacian and 
\begin{align}
\boldsymbol{\Lambda}_{\text{\rom{1}}}(k)&=\textsf{diag}\left(\left\{ \big(\mathbf{P}_i(k)+(\mathbf{H}_i^{\text{T}}\mathbf{R}_i^{-1}\mathbf{H}_i)^{-1} \big)^{-1}\right\}_{i=1}^L\right)\nonumber\\
\boldsymbol{\Lambda}_{\text{\rom{2}}}(k)&=\textsf{diag}\left(\left\{ \mathbf{\bar{M}}_i^{-1}(k)\right\}_{i=1}^L\right)\nonumber\\
\boldsymbol{\Lambda}_{\text{\rom{3}}}(k)&=\textsf{diag}\left(\left\{ \big( \mathbf{I}+\mathbf{H}_i^{\text{T}}\mathbf{R}_i^{-1}\mathbf{H}_i\mathbf{P}_i(k)\big)^{-1}\right\}_{i=1}^L\right) \nonumber
\end{align}
For an appropriate choice of $\gamma$, we have $\delta\mathbf{V}(\mathbf{\hat{e}}(k))<0$, implying that $\mathbf{\hat{e}}(k)\rightarrow\mathbf{0}$. Consequently, $\mathbf{\hat{e}}(k)=\mathbf{0}$ is asymptotically stable. Furthermore, since $\mathbf{\hat{e}}_i(k)=\mathbf{\hat{e}}_j(k)=\mathbf{0}$ for all $i\neq j$, all estimators asymptotically reach a consensus on state estimates~as $\mathbf{\hat{x}}_1(k)=~\cdots~=\mathbf{\hat{x}}_L(k)=\mathbf{x}(k)$. 

In steady-state, i.e., $k \rightarrow \infty$, we have $\boldsymbol{\Lambda}_{\text{\rom{1}}}=\lim_{k\rightarrow \infty}\boldsymbol{\Lambda}_{\text{\rom{1}}}(k)$, $\boldsymbol{\Lambda}_{\text{\rom{2}}}=\lim_{k\rightarrow \infty}\boldsymbol{\Lambda}_{\text{\rom{2}}}(k)$.
By applying statistical expectation $\mathbb{E}\{\cdot\}$  with respect to  $\mathbf{S}(k)$,  we will have the following condition for the stability of the algorithm: 
\begin{equation}\label{Condition_inq1}
\boldsymbol{\Lambda}_{\text{\rom{1}}}-\gamma^2\mathbb{E}_{}\{\mathbf{S}(k)(\mathbf{L}\otimes\mathbf{I})\boldsymbol{\Lambda}_{\text{\rom{2}}}(\mathbf{L}\otimes\mathbf{I})\mathbf{S}(k)\} \succcurlyeq 0 
\end{equation}
The expectation term can be simplified as
\begin{align}
    \mathbb{E}&_{}\{\mathbf{S}(k)(\mathbf{L}\otimes\mathbf{I})\boldsymbol{\Lambda}_{\text{\rom{2}}}(\mathbf{L}\otimes\mathbf{I})\mathbf{S}(k)\} \label{Ex_compute} \\ &=\mathbb{E}_{}\{\text{vec}^{-1}\left(\text{vec}\left(\mathbf{S}(k)(\mathbf{L}\otimes\mathbf{I})\boldsymbol{\Lambda}_{\text{\rom{2}}}(\mathbf{L}\otimes\mathbf{I})\mathbf{S}(k)\right)\right)\} \nonumber \\
    & = \mathbb{E}_{}\{\text{vec}^{-1}\left(\left(\mathbf{S}^{\text{T}}(k)\otimes\mathbf{S}(k)\right)\text{vec}\left((\mathbf{L}\otimes\mathbf{I})\boldsymbol{\Lambda}_{\text{\rom{2}}}(\mathbf{L}\otimes\mathbf{I})\right)\right)\} \nonumber \\
    & = \text{vec}^{-1}\left(\mathbb{E}_{}\{\mathbf{S}^{\text{T}}(k)\otimes\mathbf{S}(k)\}\text{vec}\left((\mathbf{L}\otimes\mathbf{I})\boldsymbol{\Lambda}_{\text{\rom{2}}}(\mathbf{L}\otimes\mathbf{I})\right)\right) \nonumber
\end{align}
Following the  approach in~\cite[Appendix~B]{PdLMS} and \cite{Vinay_IoT}, we can show that  $\mathbb{E}\{\mathbf{S}^\text{T}(k)\otimes \mathbf{S}(k)\} \leq p_e (\mathbf{I}\otimes \mathbf{I})$ with $0< p_e\leq 1$, and we have 
\begin{align}
    \mathbb{E}\{\mathbf{S}(k)&(\mathbf{L}\otimes\mathbf{I})\boldsymbol{\Lambda}_{\text{\rom{2}}}(\mathbf{L}\otimes\mathbf{I})\mathbf{S}(k)\} \nonumber \\ 
    &\leq p_e\,\text{vec}^{-1}\left(\text{vec}\left((\mathbf{L}\otimes\mathbf{I})\boldsymbol{\Lambda}_{\text{\rom{2}}}(\mathbf{L}\otimes\mathbf{I})\right)\right) \nonumber \\
    &\leq p_e\,(\mathbf{L}\otimes\mathbf{I})\boldsymbol{\Lambda}_{\text{\rom{2}}}(\mathbf{L}\otimes\mathbf{I})\cdot \nonumber
\end{align}
Subsequently, to satisfy \eqref{Condition_inq1}, the bound for $\gamma$ is determined~as 
\begin{equation}\label{gamma_simplified}
\gamma \leq \gamma^* = \frac{1}{\sqrt{p_e}} \left(\frac{\lambda_{\min}(\boldsymbol{\Lambda}_{\text{\rom{1}}})}{\lambda_{\max}((\mathbf{L}\otimes\mathbf{I})\boldsymbol{\Lambda}_{\text{\rom{2}}}(\mathbf{L}\otimes\mathbf{I}))}\right)^{\frac{1}{2}}
\end{equation}
Thus, if $\gamma$ is chosen as in~\eqref{gamma_simplified}, we can ensure that all agents reach a consensus on state estimates asymptotically, which completes the proof.
\end{proof}

\section{Resilience of the BR-CDF to Byzantine attacks}

This section investigates the robustness of the  BR-CDF in Algorithm~\ref{Alg1} to data falsification attacks. We assume that Byzantine agents start perturbing the information once the network reaches steady-state, i.e., $k = k_0>0$. We further assume that the attack remains stealthy from the perspective of agents; thus, the consensus gain $\mathbf{C}_i$ remains fixed for $k~\geq~k_0$. 

In steady-state, the error covariance matrix in~\eqref{Pi_BRDKF} satisfies 
\begin{align}
\mathbf{P}_i=&\mathbf{\hat{F}}_i\mathbf{P}_i\mathbf{\hat{F}}_i^\text{T} +\mathbf{K}_i\mathbf{R}_i\mathbf{K}_i^\text{T} +\mathbf{Q}\label{f21}\\
&+ \mathbf{C}_i   \mathbb{E}_{}\left\{\textstyle\sum_{s \in \mathcal{N}_i}  \textstyle\sum_{p \in \mathcal{N}_i}\mathbf{S}_s(k) \boldsymbol{\Sigma}_{sp} \mathbf{S}_p^\text{T}(k)\right\}\mathbf{C}_i^\text{T} \nonumber
\end{align}
where $\mathbf{P}_i=\lim_{k \rightarrow \infty}\mathbf{P}_i(k)$.
Defining  $\mathbf{P}\triangleq\textsf{diag}(\{ \mathbf{P}_i\}_{i=1}^L)$, $\mathbf{\hat{F}}\triangleq\textsf{diag}(\{ \mathbf{\hat{F}}_i\}_{i=1}^L)$, $\mathbf{K}\triangleq\textsf{diag}(\{ \mathbf{K}_i\}_{i=1}^L)$,  $\mathbf{C}\triangleq\textsf{diag}(\{ \mathbf{C}_i\}_{i=1}^L)$, and $\mathbf{R}\triangleq\textsf{diag}(\{ \mathbf{R}_i\}_{i=1}^L)$,  the network-wide version of \eqref{f21} is given by 
\begin{align}
&\mathbf{P}=\mathbf{\hat{F}}\mathbf{P}\mathbf{\hat{F}}^\text{T} +\mathbf{K}\mathbf{R}\mathbf{K}^\text{T} +\mathbf{I}_L \otimes \mathbf{Q}\label{f21_Net}\\
&+ \mathbf{C}  \mathbb{E}_{}\left\{(\mathbf{I}_L \otimes \mathbf{I})\odot\left((\mathbf{E}\otimes \mathbf{I})\mathbf{S}(k)\mathbf{\Sigma}\mathbf{S}^{\text{T}}(k)(\mathbf{E}\otimes \mathbf{I})\right)\right\}\mathbf{C}^\text{T} \nonumber
\end{align}
where $\mathbf{\Sigma}$ is the network-wide coordinated attack covariance. Under \textbf{Assumption~1}, we have 
\begin{align}
&\mathbf{P}=\mathbf{\hat{F}}\mathbf{P}\mathbf{\hat{F}}^\text{T} +\mathbf{K}\mathbf{R}\mathbf{K}^\text{T} +\mathbf{I}_L \otimes \mathbf{Q}\label{f21_Net_1}\\
&+ \mathbf{C}  \bigg((\mathbf{I}_L \otimes \mathbf{I})\odot\left((\mathbf{E}\otimes \mathbf{I})\mathbb{E}_{}\left\{\mathbf{S}(k)\mathbf{\Sigma}\mathbf{S}^{\text{T}}(k)\right\}(\mathbf{E}\otimes \mathbf{I})\right)\bigg)\mathbf{C}^\text{T} \nonumber
\end{align}
and  using the result of~\eqref{Ex_compute}, we finally have 
\begin{align}
\mathbf{P} = &\,\mathbf{\hat{F}}\mathbf{P}\mathbf{\hat{F}}^\text{T} +\mathbf{K}\mathbf{R}\mathbf{K}^\text{T} +\mathbf{I}_L \otimes \mathbf{Q}\label{f21_Net_2}\\
&+ p_e \mathbf{C}  \bigg( (\mathbf{I}_L \otimes \mathbf{I})\odot\big((\mathbf{E}\otimes \mathbf{I})\mathbf{\Sigma}(\mathbf{E}\otimes \mathbf{I})\big)\bigg)\mathbf{C}^\text{T} \nonumber
\end{align}
The last term in~\eqref{f21_Net_2} describes the impact of the coordinated Byzantine attack on the error covariance matrix that is scaled by $p_e$.
Thus, defining the steady-state  network-wide MSE (NMSE) as
\begin{equation}
 \label{Net_MSE}
 \text{NMSE}\triangleq \lim_{k \rightarrow \infty}\text{tr}(\mathbb{E}\{\mathbf{P}(k)\}),
  \end{equation}
we see that partial sharing, i.e., $p_e < 1$, results in lower steady-state NMSE compared to the case when the full state is shared, i.e., $p_e=1$, which gives enhanced robustness against coordinated Byzantine attacks.

\section{Coordinated Byzantine Attack Design}\label{coor_Byz_att}
To analyze the worst-case performance of the BR-CDF algorithm, we consider a scenario where Byzantine agents design a coordinated attack to maximize the NMSE. Based on the  attack model in Section~\ref{Sys_Model_DKF} and the  error covariance of  the  BR-CDF algorithm in~\eqref{f6}, Byzantine agents have the following two levers to design their coordinated attack:
\begin{itemize}
    \item The design of perturbation covariance matrix $\mathbf{\Sigma}$, modeled as the covariance of zero-mean Gaussian sequences.
    \item The choice of selection matrices that impacts the sequence of information fractions that Byzantine agents share at the beginning of the attack, i.e., $\mathbf{S}_i(k_0)$ for $i \in \mathcal{B}$.
\end{itemize}

We ensure that the attack remains stealthy from the perspective of regular agents by setting an upper bound on the energy of the perturbation sequences, i.e., $\text{tr}(\boldsymbol{\Sigma}) \leq \eta$. Assuming Byzantines start perturbing information once agents reach steady-state, i.e., $k = k_0$, we derive an expression for the NMSE pertaining to the estimator in~\eqref{f5}. The network-wide evolution of the estimation error of the BR-CDF algorithm,  given in~\eqref{Error_actual}, is given by
\begin{equation}
	\mathbf{e}(k+1)=\mathbf{\tilde{A}}(k)\mathbf{e}(k)+\mathbf{\tilde{b}}(k)+\boldsymbol{\Gamma}(k){\boldsymbol{\delta}}(k)
\end{equation}
where $\mathbf{e}(k)\triangleq[\mathbf{e}_1^\text{T}(k), \cdots,  \mathbf{e}_L^\text{T}(k) ]^\text{T}$,
\begin{align}
\mathbf{\tilde{A}}(k)&=\textsf{diag}(\{ \mathbf{F}_i(k)\}_{i=1}^L)+\mathbf{C}\big(\mathbf{E}\otimes\mathbf{I}\big)\mathbf{S}(k)\nonumber\\
\mathbf{\tilde{b}}(k)&=\textsf{diag}(\{ \mathbf{{K}}_i(k)\mathbf{v}_i(k)\}_{i=1}^L)-\mathbf{1}_L\otimes \mathbf{w}(k)\nonumber\\
\boldsymbol{\Gamma}(k)&=\mathbf{C}\big(\mathbf{E}\otimes\mathbf{I}\big)\mathbf{S}(k) \nonumber
\end{align}
As a result, the network-wide error covariance matrix $\mathbf{P}(k) \triangleq \mathbb{E}\{\mathbf{e}(k)\mathbf{e}^\text{T}(k)\}$,  including cross-terms of the error covariance, is given by
\begin{equation}
\label{Covar_actual}
\mathbf{P}(k+1)= \mathbf{\tilde{A}}(k)\mathbf{P}(k)\mathbf{\tilde{A}}^\text{T}(k) + \mathbf{\tilde{Q}}(k) + \boldsymbol{\Gamma}(k) \boldsymbol{\Sigma} \boldsymbol{\Gamma}^\text{T}(k)
\end{equation}
where $\mathbf{\tilde{Q}}(k)= \textsf{diag}(\{ {\mathbf{K}}_i(k)\mathbf{R}_i {\mathbf{K}}_i^\text{T}(k)\}_{i=1}^L  + \mathbf{1}_L \mathbf{1}_L^\text{T} \otimes \mathbf{Q}$.
In~\eqref{Covar_actual}, the last term is due to the injected noise and is given~by 
\begin{align}
    \boldsymbol{\Gamma}(k) \boldsymbol{\Sigma}  \boldsymbol{\Gamma}^\text{T}(k)=\mathbf{C}\big(\mathbf{E}\otimes\mathbf{I}\big)\mathbf{S}(k)\boldsymbol{\Sigma}\mathbf{S}^\text{T}(k)\big(\mathbf{E}\otimes\mathbf{I}\big)\mathbf{C}^\text{T} \label{Deg_fac}
\end{align}
 which, compared to the Byzantine-free case, degrades the NMSE.
Considering the NMSE in~\eqref{Net_MSE}, we define two optimization problems to find the optimal coordinated Byzantine attacks by designing the partial-sharing selection matrices at $k=k_0$ and attack covariance matrices of Byzantine agents.

The last term of the estimation error covariance $\mathbf{P}(k)$, as in~\eqref{Deg_fac},  is the only term of~\eqref{Covar_actual} that depends on the attack; thus, maximizing the trace of the estimation error covariance is equivalent to maximizing the trace of its last term \cite{guo2018worst}. The last term of  $\mathbf{P}(k)$ in~\eqref{Covar_actual} also depends on the selection matrix $\mathbf{S}(k)$ and given the attack  covariance $\boldsymbol{\Sigma}$, we can show that  
\begin{align}
     \text{tr}(\boldsymbol{\Gamma}(k) \boldsymbol{\Sigma}  \boldsymbol{\Gamma}^\text{T}(k))&=\text{tr}\big(\mathbf{C}(\mathbf{E}\otimes\mathbf{I})\mathbf{S}(k)\boldsymbol{\Sigma}\mathbf{S}^\text{T}(k)(\mathbf{E}\otimes\mathbf{I})\mathbf{C}^\text{T}\big)\nonumber \\
     &= \text{tr}\big((\mathbf{E}\otimes\mathbf{I})\mathbf{C}^\text{T}\mathbf{C}(\mathbf{E}\otimes\mathbf{I})\mathbf{S}(k)\boldsymbol{\Sigma}\mathbf{S}^\text{T}(k)\big) \nonumber\\
     &=\textstyle\sum_{i \in \mathcal{B}} \textstyle\sum_{j \in \mathcal{B}} \text{tr} \left(\mathbf{U}_{ij} \mathbf{S}_j(k) \boldsymbol{\Sigma}_{ji} \mathbf{S}_i(k)\right)\label{tr_last_term}
 \end{align} 
 where $\mathbf{U}_{ij}= \textstyle\sum_{q \in \mathcal{N}_i \cap\mathcal{N}_j} \mathbf{C}_q^\text{T}\mathbf{C}_q$. Thus, the optimization problem that maximizes the steady-state NMSE can be stated as
 \begin{equation}
	\begin{aligned}\label{OPT_S}
	&& \underset{\{\mathbf{S}_i,\, i \in \mathcal{B}\} }{\max} 
	 &  \sum_{i \in \mathcal{B}} \sum_{j \in \mathcal{B}} \text{tr} \left(\mathbf{U}_{ij} \mathbf{S}_j \boldsymbol{\Sigma}_{ji} \mathbf{S}_i\right) \\
	&& \textrm{s. t.}\,\,\,
	 &  \mathbf{0} \leq\mathbf{S}_i \leq \mathbf{I} \quad \forall i \in \mathcal{B} \\
	 &&& [\mathbf{S}_i]_{rs} \in \{0,1\} \\
	 &&& \text{tr}(\mathbf{S}_i) \leq l \quad \forall i \in \mathcal{B}
	\end{aligned}
\end{equation}
where the resulted solution for $\mathbf{S}_i$ determines the  $\mathbf{S}_i(k_0)$ and the first two constraints restrict the selection matrix to be diagonal with $0$ or $1$ elements on the main diagonal.
The last constraint enforces that only $l$ elements of the state vector are shared with neighbors at each given instant. 
We relax the non-convex Boolean constraint on the elements of $\mathbf{S}_i$ and rewrite the optimization problem as
\begin{equation}
	\begin{aligned}\label{OPT_S_main}
	&& \underset{\{\mathbf{S}_i, \, i \in \mathcal{B}\} }{\max} 
	 &  \sum_{i \in \mathcal{B}} \sum_{j \in \mathcal{B}} \text{tr} \left(\mathbf{U}_{ij} \mathbf{S}_j \boldsymbol{\Sigma}_{ji} \mathbf{S}_i\right) \\
	&& \textrm{s. t.}\,\,\,
	 &  \mathbf{0} \leq\mathbf{S}_i \leq \mathbf{I} \quad \forall i \in \mathcal{B} \\
	 &&& \text{tr}(\mathbf{S}_i) \leq l \quad \forall i \in \mathcal{B}
	\end{aligned}
\end{equation}

The objective function  in~\eqref{OPT_S_main} can be  further simplified as 
\begin{align}
    \label{eq24}
    \sum_{i \in \mathcal{B}} \sum_{j \in \mathcal{B}} & \text{tr} \left(\mathbf{U}_{ij} \mathbf{S}_j \boldsymbol{\Sigma}_{ji} \mathbf{S}_i\right) =\sum_{i \in \mathcal{B}}\bigg(\text{tr} \left(\mathbf{U}_{ii} \mathbf{S}_i \boldsymbol{\Sigma}_{i} \mathbf{S}_i\right)\\
    &+\sum_{j \in \mathcal{B}/\{i\}} \frac{1}{2}\text{tr} \left(\mathbf{U}_{ij} \mathbf{S}_j \boldsymbol{\Sigma}_{ji} \mathbf{S}_i +\mathbf{U}_{ji} \mathbf{S}_i \boldsymbol{\Sigma}_{ij} \mathbf{S}_j\right)\bigg)\nonumber
\end{align}
which still contains non-convex quadratic terms. To overcome this problem, we employ the block coordinate descent (BCD) algorithm where each Byzantine agent $i$, given the selection matrix of other Byzantines,  optimizes its own selection matrix. The BCD algorithm is iterated for $T$ iterations and at each iteration $t+1$, agent $i$ employs the selection matrix of other Byzantine agents from the previous iteration, i.e.  $\{\mathbf{S}_j^{t}\}_{j \in \mathcal{B} \backslash \{i\}}$.

Hence, the optimization problem in~\eqref{OPT_S_main} can be solved by employing the BCD method, where at each agent $i \in \mathcal{B}$ and BCD iteration $t+1$, the optimization problem is modeled as 
\begin{equation}
\label{BCD_opt}
	\begin{aligned}
	\mathcal{P}: \quad & \underset{\mathbf{S}_i }{\max} 
	 && F(\mathbf{S}_i,\{\mathbf{S}_j^t\}_{j \in \mathcal{B}\backslash \{i\}}) \\
	& \textrm{s. t.}\,\,\,
	 &&  \mathbf{0} \leq\mathbf{S}_i \leq \mathbf{I}\\
	 &&& \text{tr}(\mathbf{S}_i) \leq l
	\end{aligned}
\end{equation}
with the objective function
\begin{align}
F(\mathbf{S}_i,&\{\mathbf{S}_j^t\}_{j \in \mathcal{B}  \backslash \{i\}})= \text{tr} \left(\mathbf{U}_{ii} \mathbf{S}_i \boldsymbol{\Sigma}_{i} \mathbf{S}_i\right)\label{eq23}\\
&+\sum_{j \in \mathcal{B}/\{i\}} \frac{1}{2}\text{tr} \left(\mathbf{U}_{ij} \mathbf{S}_j^t \boldsymbol{\Sigma}_{ji} \mathbf{S}_i +\mathbf{U}_{ji} \mathbf{S}_i \boldsymbol{\Sigma}_{ij} \mathbf{S}_j^t\right) \nonumber
\end{align}
and $\mathbf{S}_j^t$ as the selection matrix of  Byzantine agent $j$ at the former BCD iteration.
\begin{algorithm}[t]
\caption{BCD-based attack design}
\label{Alg2}
\begin{algorithmic}
\renewcommand{\algorithmicrequire}{ For}
\REQUIRE each agent $i \in \mathcal{B}$
\renewcommand{\algorithmicrequire}{ Receive}
\REQUIRE $\mathbf{S}_j^0=\mathbf{S}_j(k_0)$ from $j \in \mathcal{B}\backslash \{i\}$
\renewcommand{\algorithmicrequire}{ Share}
\REQUIRE $\mathbf{S}_i^0=\mathbf{S}_i(k_0)$ with $j \in \mathcal{B}\backslash \{i\}$
\FOR{$t=1$ \TO T}
\STATE Find $\mathbf{S}_i$ by solving $\mathcal{P}$ in~\eqref{BCD_opt}
\STATE Set $\mathbf{S}_i^t=\mathbf{S}_i$ and share with $j \in \mathcal{B}\backslash \{i\}$
\STATE Receive $\{\mathbf{S}_j^t\}_{j \in \mathcal{B}\backslash \{i\}}$
\ENDFOR
\STATE For the main diagonal of the $\mathbf{S}_i^T$, set the $l$ largest element to $1$ and the others to $0$.
\STATE Set $\mathbf{S}_i(k_0)=\mathbf{S}_i^T$
\end{algorithmic} 
\end{algorithm}
 Algorithm~\ref{Alg2} summarizes the BCD algorithm used to solve the optimization problem in~\eqref{BCD_opt}. Next, we investigate how optimizing the perturbation covariance matrix impacts the NMSE.

Given the selection matrices at the beginning of the attack, i.e., $\mathbf{S}_i(k_0)$ for $i \in \mathcal{N}$, Byzantine agents can maximize the steady-state NMSE  by cooperatively designing their attack covariances in the following optimization problem
 \begin{equation}
	\begin{aligned}\label{OPT}
	&& \underset{\boldsymbol{\Sigma} }{\max} 
	 &\quad  \text{tr}(\boldsymbol{\Gamma}(k_0) \boldsymbol{\Sigma} \boldsymbol{\Gamma}^{\text{T}}(k_0)) \\
	&& \textrm{s. t.}\,\,\,
	 &   \quad\boldsymbol{\Sigma} \succcurlyeq 0 \\
	 &&& \quad\text{tr}(\boldsymbol{\Sigma}) \leq \eta
	\end{aligned}
\end{equation}
where $\small{\boldsymbol{\Gamma}(k_0)=\mathbf{C}\big(\mathbf{E}\otimes\mathbf{I}\big)\mathbf{S}(k_0) (\,\textsf{diag}(\mathbf{z}) \otimes \mathbf{I})}$ and  $\mathbf{z} \in \mathbb{R}^L$ is a vector designed to preserve the structure of the perturbation covariance. We introduce the Boolean vector $\mathbf{z}=[z_1, z_2,\cdots,z_L]^\text{T}$ where $z_i=1$ if $i \in \mathcal{B}$ and $z_i=0$ otherwise. By employing  $\mathbf{z}$, the block matrices of the $\boldsymbol{\Sigma}$ that correspond to regular agents are all set to zero. The first constraint in~\eqref{OPT} guarantees that the designed attack covariance is positive semidefinite and the last constraint is related to stealthiness. The energy of the Byzantine noise sequences is assumed to be limited as $\text{tr}(\boldsymbol{\Sigma}) \leq \eta$ to maintain the attack stealthiness. 

\begin{rem}
The optimization problem in \eqref{OPT} is a semidefinite programming (SDP) problem that can be efficiently solved by interior-point methods.
\end{rem}

\section{Simulation Results}
In this section, we demonstrate the robustness of the  BR-CDF algorithm to Byzantine attacks. For this purpose, we consider  a target tracking problem with the state vector length of $m=8$, described by a linear model  
\begin{equation*}
	\mathbf{x}(k + 1) = \left(\begin{bmatrix}
    0.6&0.005 \\ 0.25&0.6
    \end{bmatrix} \otimes \mathbf{I}_4\right)\, \mathbf{x}(k) +\mathbf{w}(k)
\end{equation*}
\begin{figure}[!t]
	\centering
	\includegraphics[width=.4\textwidth]{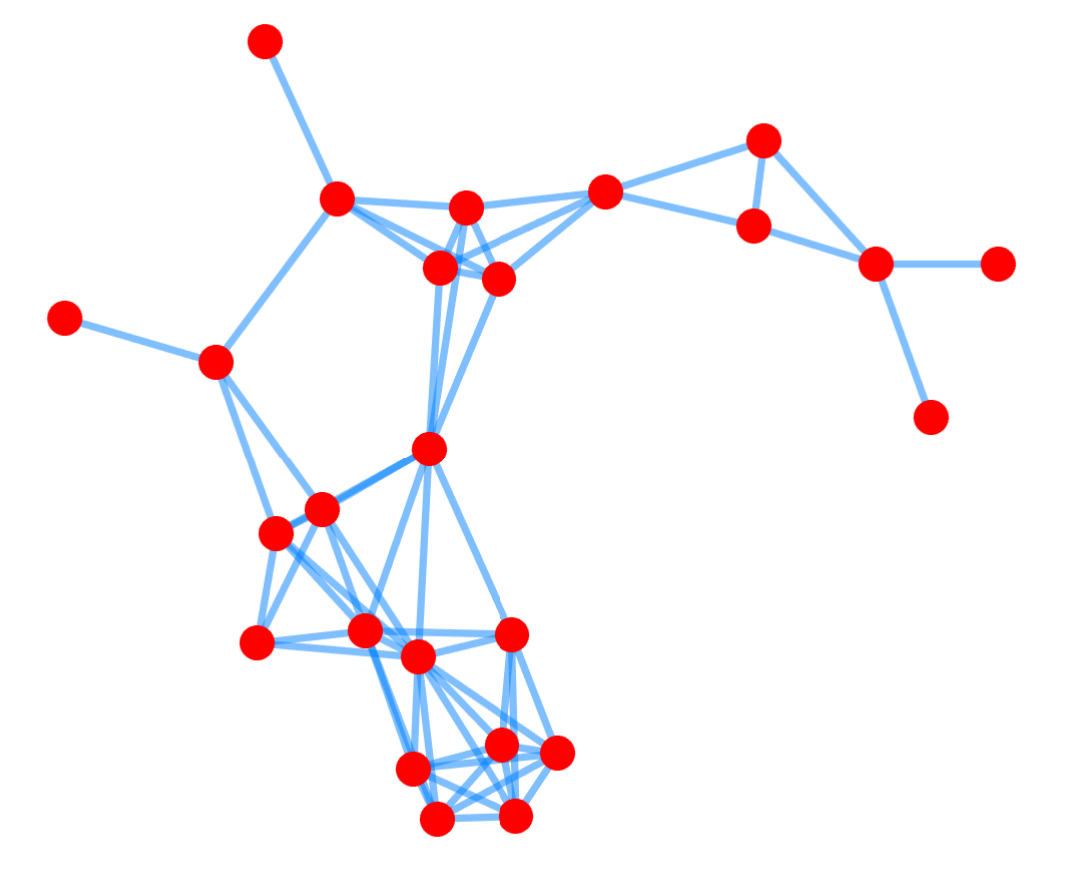}
	\caption{Randomly generated network topology.}
	\label{Fig_1}
\end{figure}
We considered a randomly generated undirected connected network with $L=25$ agents,  as shown in Fig.~\ref{Fig_1}. At each agent $i$, the state noise covariance is $\mathbf{Q}=0.1\mathbf{I}$ and the local observation is given by 
\begin{equation*}
\mathbf{y}_i(k) = \left(\begin{bmatrix}
    1&1&0&0 \\ 1&0&0&0\\0&0&1&0\\0&0&1&1
    \end{bmatrix} \otimes \mathbf{I}_2\right)\,\mathbf{x}(k) + \mathbf{v}_i(k)
\end{equation*}
In addition, at each agent $i$, we considered the observation noise covariance as $\mathbf{R}_i=\mu_i \mathbf{I}$, where $\mu_i \sim \mathcal{U}(0,1)$.  The average NMSE of agents is considered as a performance metric, i.e., 
\begin{equation}
\label{MSE}
\text{MSE}\triangleq\frac{1}{L}\textstyle\sum_{i=1}^L \text{tr}(\mathbf{P}_i)
\end{equation}
with $\mathbf{P}_i$ being the steady-state error covariance matrix of agent $i$ in~\eqref{f21}. The simulation results presented in the following are obtained by averaging over 100 independent experiments.
\begin{figure}[!t]
	\centering
	\includegraphics[width=.5\textwidth]{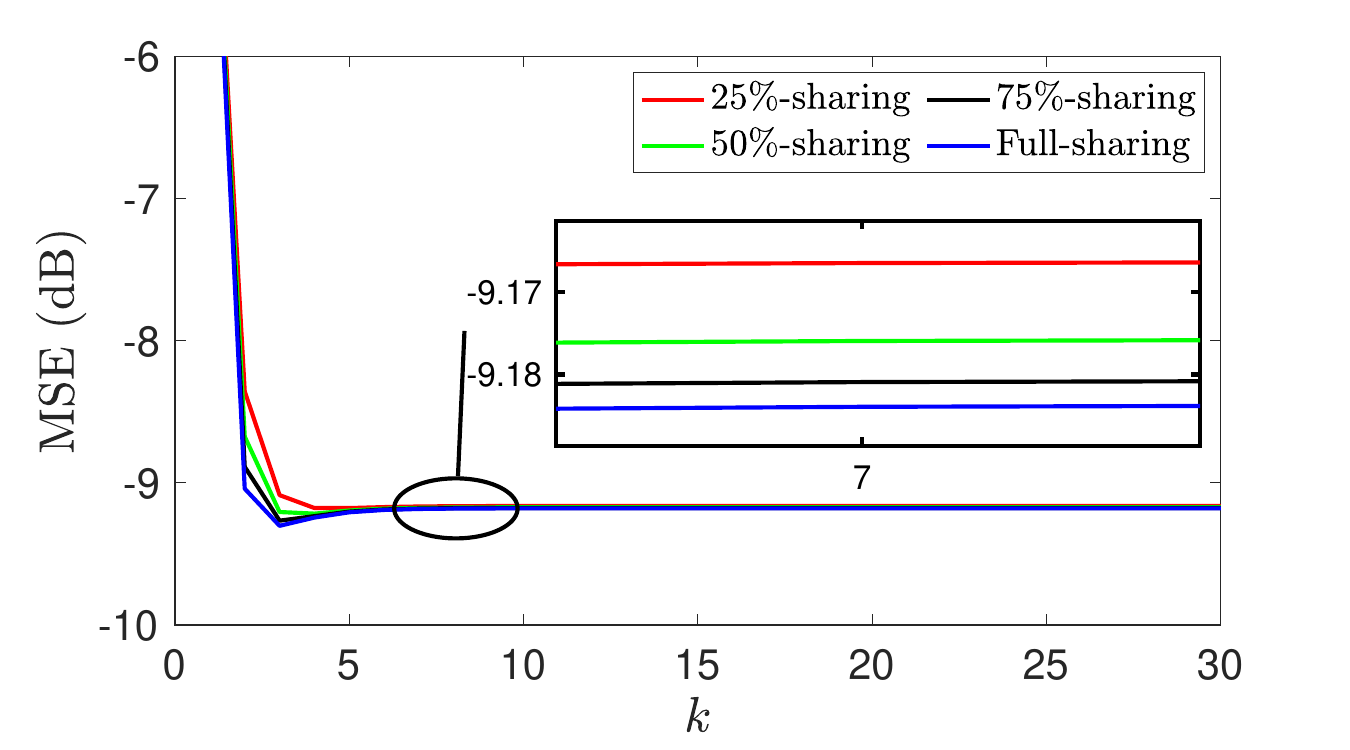}
	\caption{MSE of the  BR-CDF  algorithm versus time index $k$ without~attack.}
	\label{Fig_1_2}
\end{figure}

We simulated the proposed BR-CDF algorithm for different values of $l$, e.g., 2, 4, 6, and 8 (i.e., $25\%$, $50\%$, $75\%$ and $100\%$ information sharing). Fig.~\ref{Fig_1_2} shows the corresponding learning curves, i.e., MSE versus time instant $k$, when no attacks occur in the network. We see that the performance degradation is inversely proportional to the amount of information sharing. Although sharing a smaller fraction results in higher MSE, the difference is negligible in this experiment.

\begin{figure}[!t]
	\centering
	\includegraphics[width=.5\textwidth]{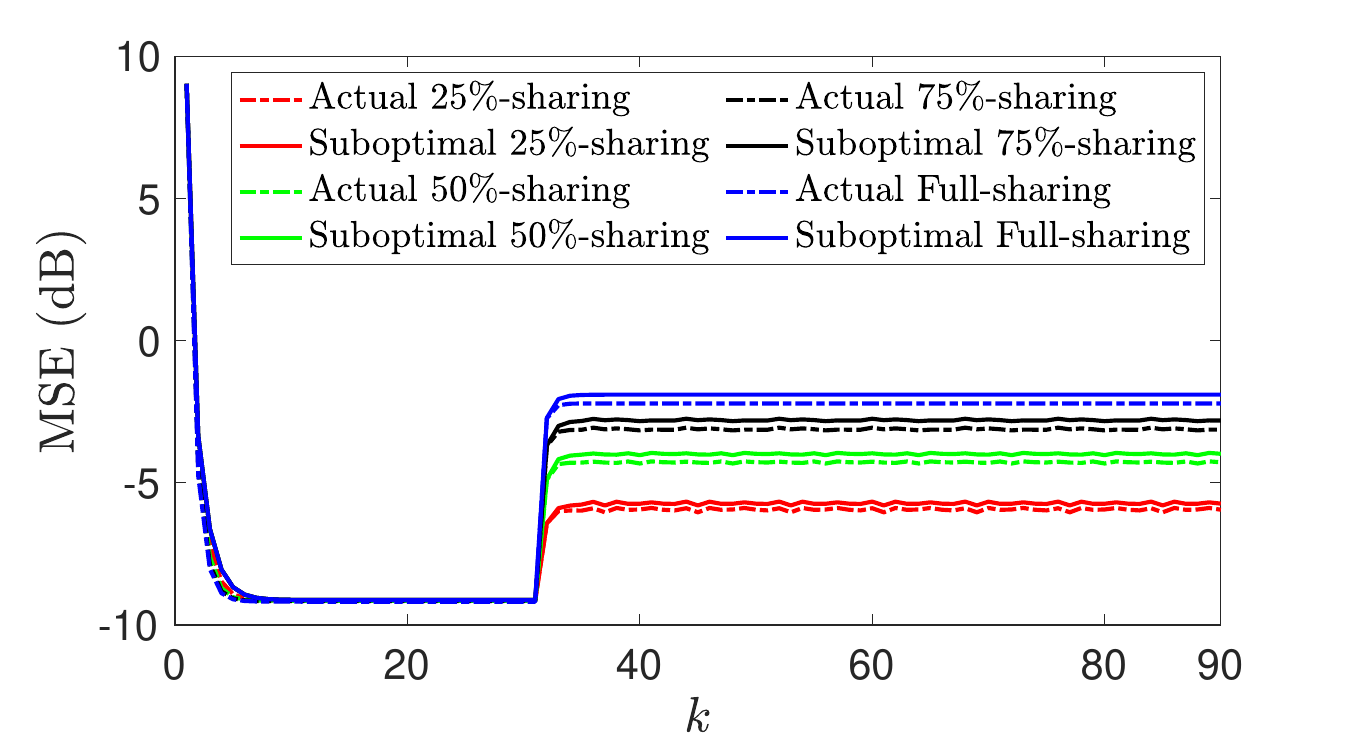}
	\caption{MSE of the BR-CDF and its suboptimal solution versus time index~$k$.}
	\label{Fig_2}
\end{figure}

Next, we examined the robustness of the  BR-CDF algorithm to Byzantine attacks. After the network has reached convergence, the Byzantine agents launch an attack at $k_0=30$. The Byzantine agents are chosen as the $B=5$ nodes with the highest degree in the network graph and the energy of the attack sequences is restricted with parameter $\eta=L$.  We then compared the accuracy of the proposed suboptimal BR-CDF in Algorithm~\ref{Alg1} to the solution of the BR-CDF that
shares all necessary variables. Fig.~\ref{Fig_2} illustrates their corresponding learning curves for different values of $l$. We observe that the suboptimal solution performs closely to the solution that
shares all necessary variables. Furthermore, the proposed algorithms provide robustness against Byzantine attacks since sharing less information results in lower MSE.

In Fig.~\ref{Fig_3}, in order to observe the fluctuation caused by the selection matrices, we plot the MSE in~\eqref{MSE}  and  $\text{MSE}^\prime=\frac{1}{L}\lim_{k \rightarrow \infty}\sum_{i=1}^L \text{tr}(\mathbf{P}_i(k))$ with $\mathbf{P}_i(k)$ in \eqref{Pi_BRDKF}.\footnote{The difference between $\text{MSE}^\prime$ and MSE is that the $\text{MSE}^\prime$ does not include the statistical expectation with respect to the selection matrices.} Thus, we can examine the accuracy of our theoretical finding to compute the expected value of the error covariance with respect to the selection matrices in \eqref{f21}. The comparable performance of the MSE and $\text{MSE}^\prime$ in Fig.~\ref{Fig_3} demonstrates that simulation results match theoretical findings.

\begin{figure}[!t]
	\centering
	\includegraphics[width=.5\textwidth]{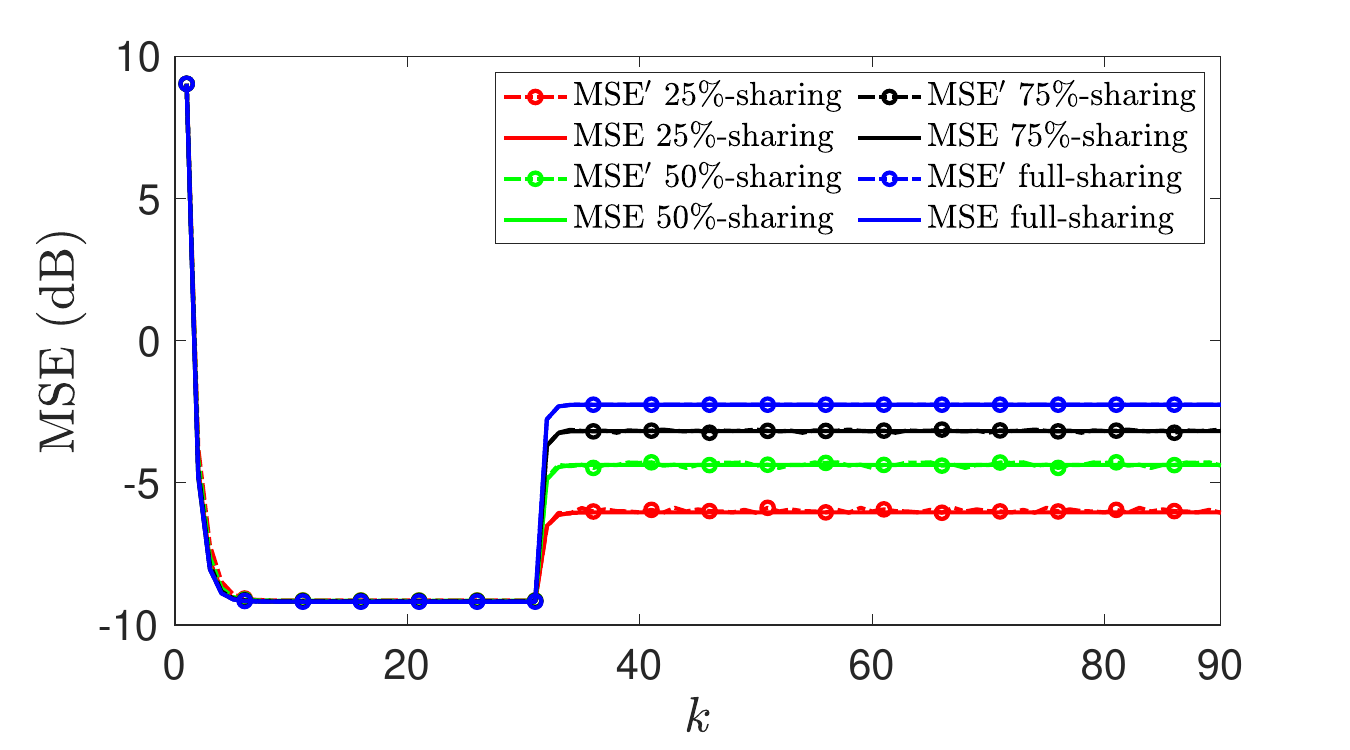}
	\caption{MSE and $\text{MSE}^\prime$ versus  time index $k$.}
	\label{Fig_3}
\end{figure}

\begin{figure}[!t]
	\centering
	\includegraphics[width=.5\textwidth]{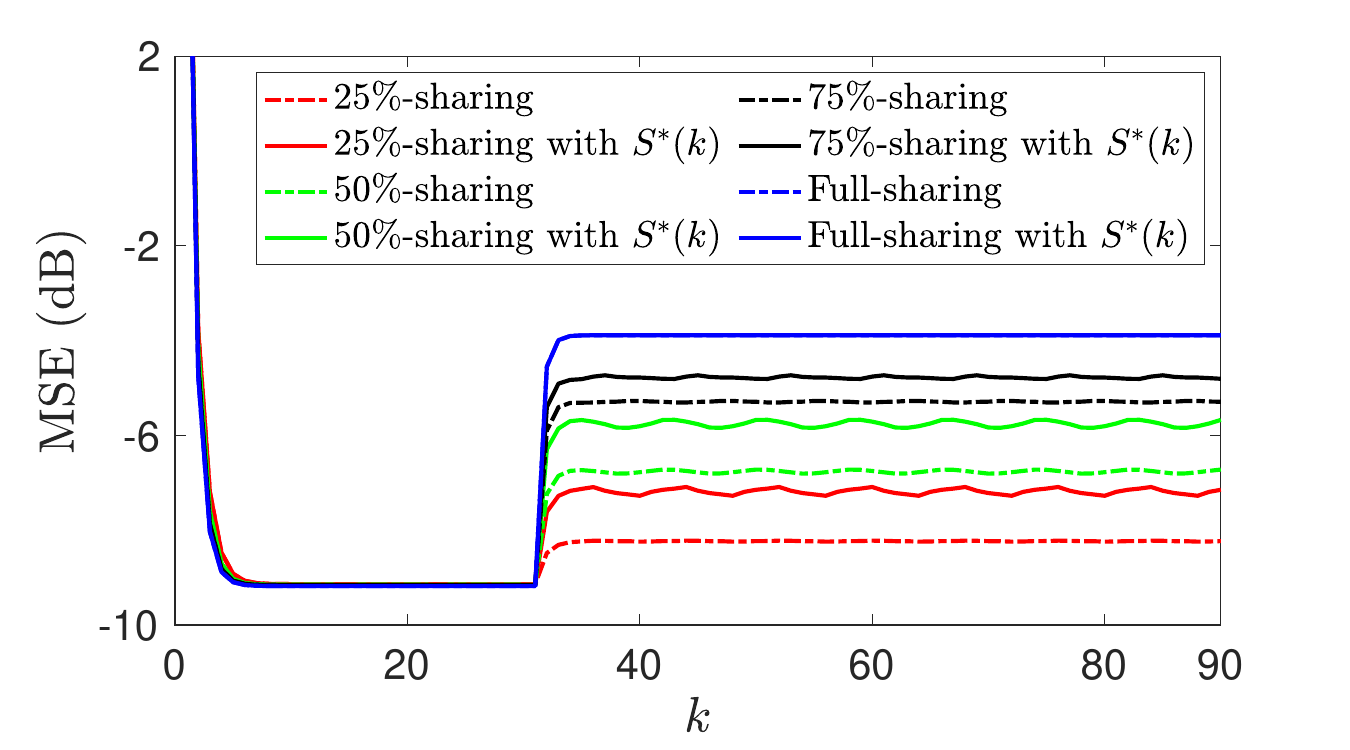}
	\caption{MSE versus  time index $k$ for cases of optimized selection matrix $\mathbf{S}^*(k)$ and random  selection matrix $\mathbf{S}(k)$.}
	\label{Fig_5}
\end{figure}

To solve the optimization problem $\mathcal{P}$ in~\eqref{BCD_opt}, we performed the simulation by the BCD algorithm with $T=10$ iterations and designed the selection matrices $\{\mathbf{S}_j(k_0)\}_{j \in \mathcal{B}}$  at $k_0=30$. We can see from Fig.~\ref{Fig_5} that the designed selection matrix increases the network MSE. Also, it can be seen that designing the selection matrices has a higher impact on the network performance when a smaller fraction is shared. 

\begin{figure}[!t]
	\centering
	\includegraphics[width=.5\textwidth]{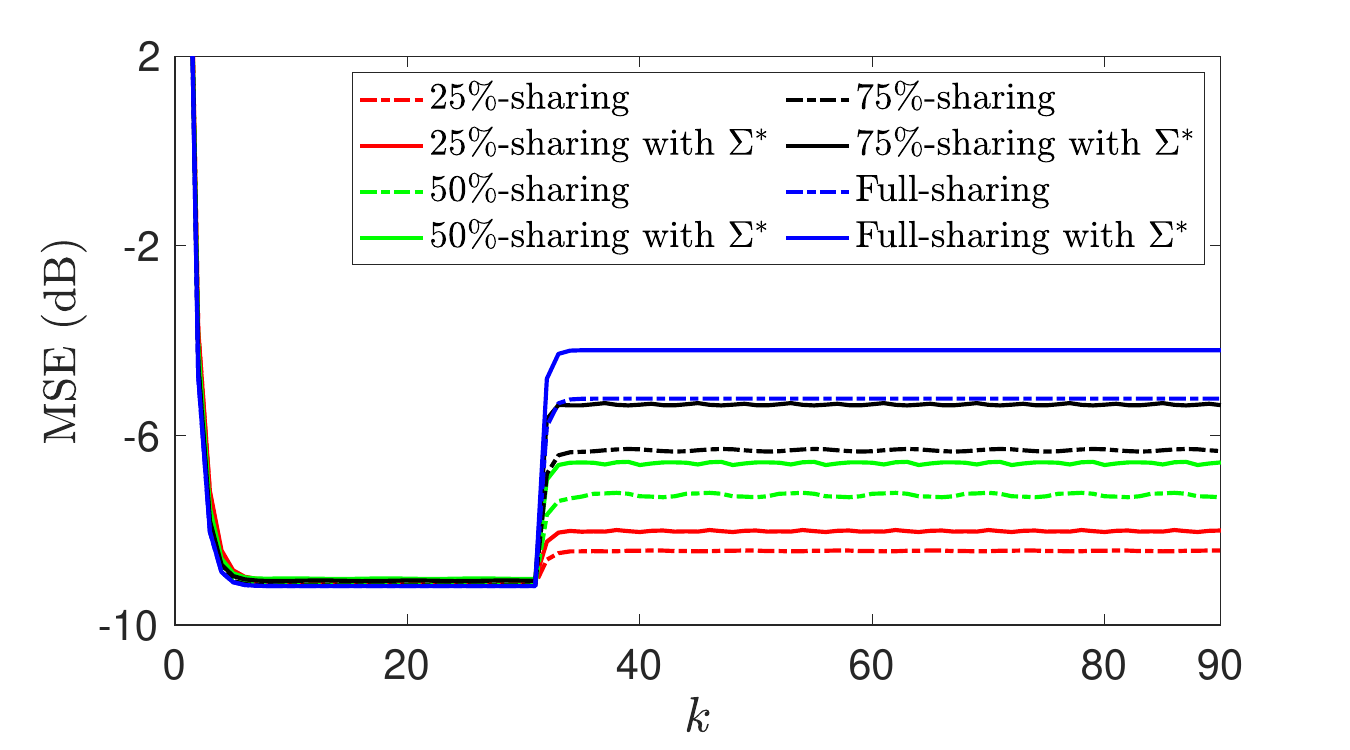}
	\caption{MSE versus  time index $k$ for cases of optimized attack covariance $\mathbf{\Sigma}^*$ and random  attack covariance $\mathbf{\Sigma}$.}
	\label{Fig_6}
\end{figure}

\begin{figure*}[!t] 
\centering
\subfloat{\includegraphics[width=.265 \linewidth]{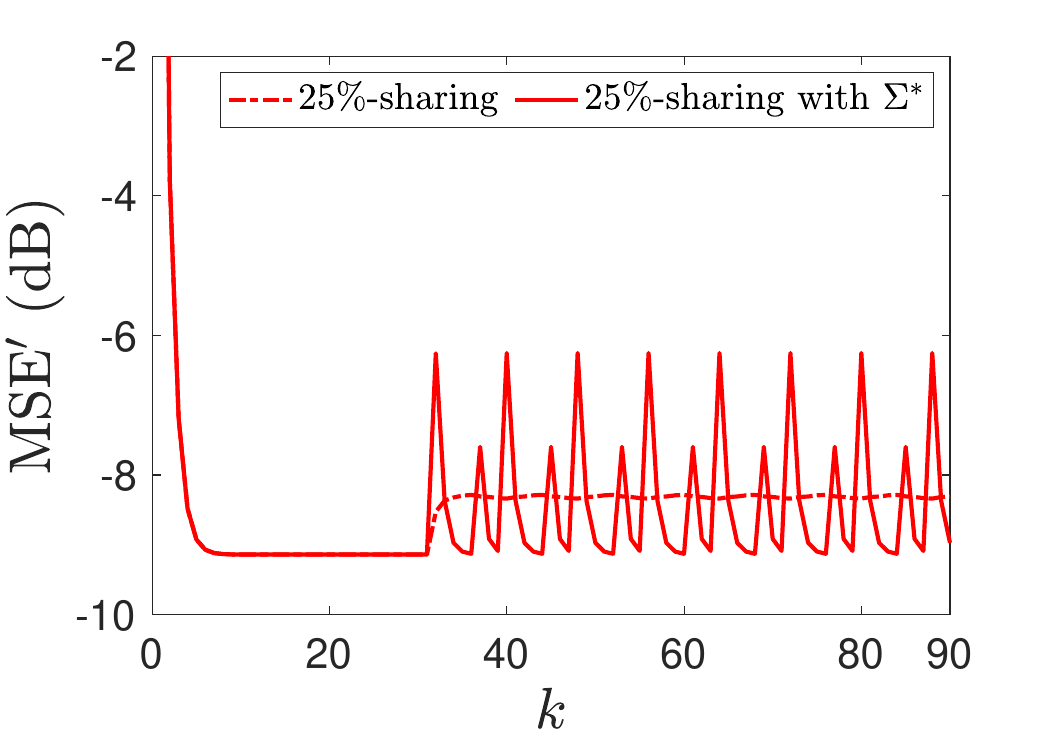}}\hspace{-4.8mm}
\subfloat{\includegraphics[width=.265 \linewidth]{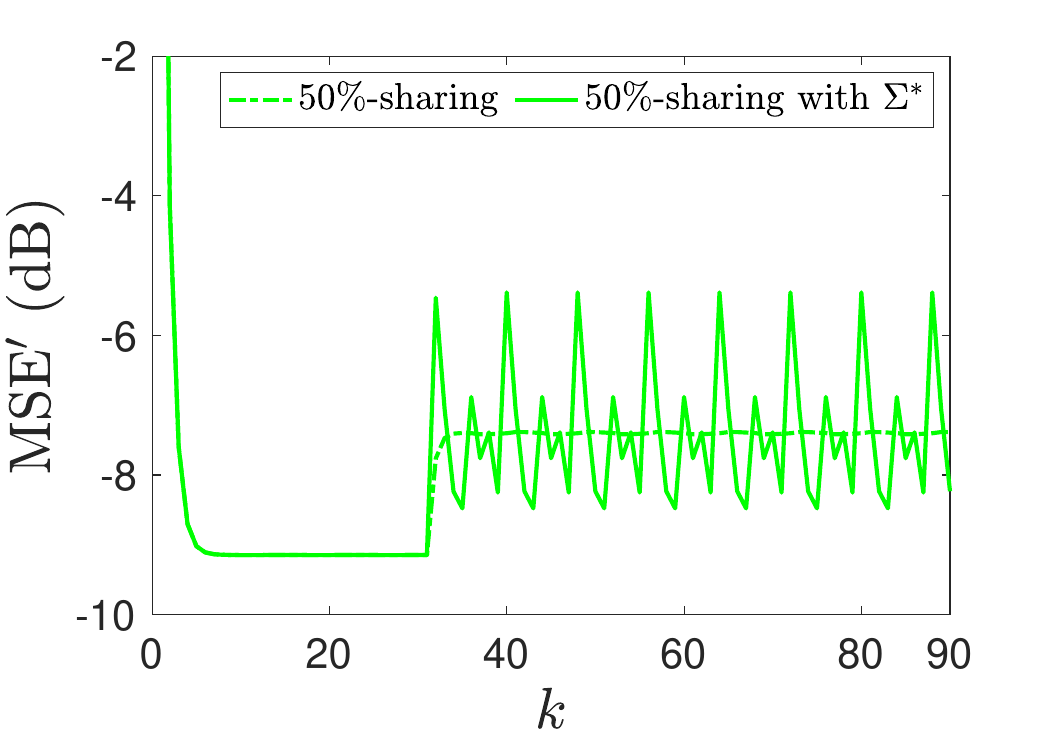}}\hspace{-4.8mm}
\subfloat{\includegraphics[width=.265 \linewidth]{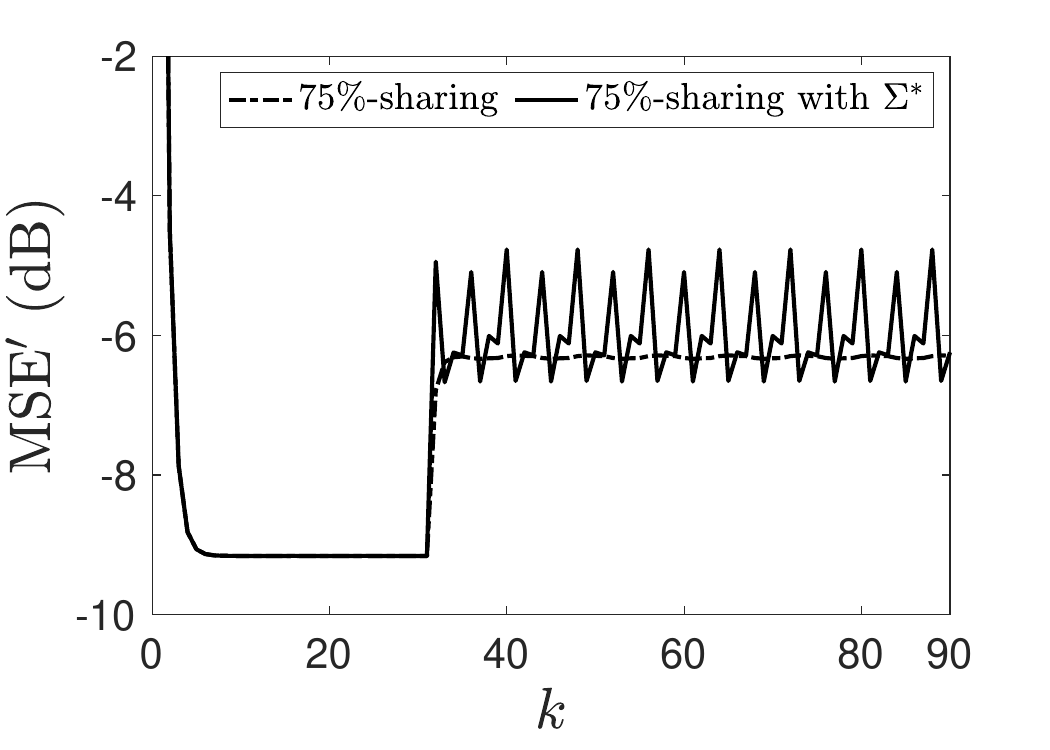}}\hspace{-4.8mm}
\subfloat{\includegraphics[width=.265 \linewidth]{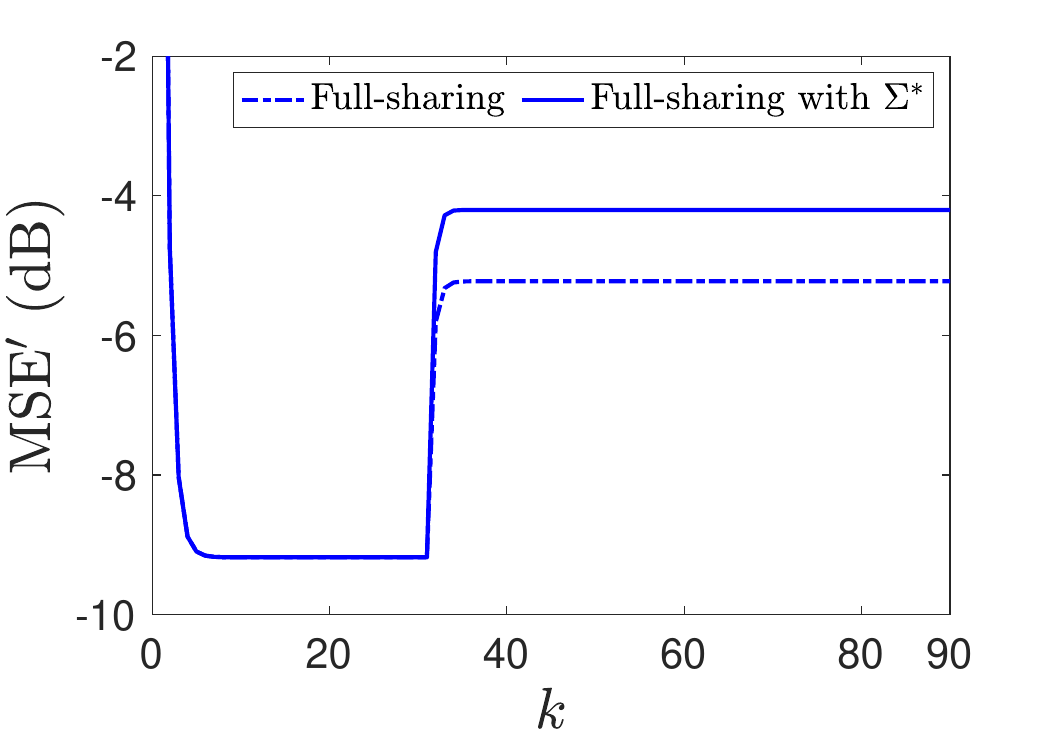}}\hspace{-4.8mm}
\caption{$\text{MSE}^\prime$ versus  time index $k$ for cases of optimized attack covariance  $\mathbf{\Sigma}^*$ and random  attack covariance $\mathbf{\Sigma}$.}
\label{Fig_7}
\end{figure*}

By solving the optimization problem in~\eqref{OPT}, we examine the impact of optimizing the attack covariance compared to a random attack covariance. To this end, we fixed the constraint on the energy of the perturbation sequences, i.e., $\eta$. Fig.~\ref{Fig_6} shows that optimizing the perturbation covariance $\mathbf{\Sigma}^*$ increases the MSE, while using partial sharing of information enhances robustness to Byzantine attacks by restricting the growth in MSE. In other words, as we share more information with neighbors, the impact of optimizing the perturbation covariance matrix increases.

\begin{figure}[!t]
	\centering
	\includegraphics[width=.5\textwidth]{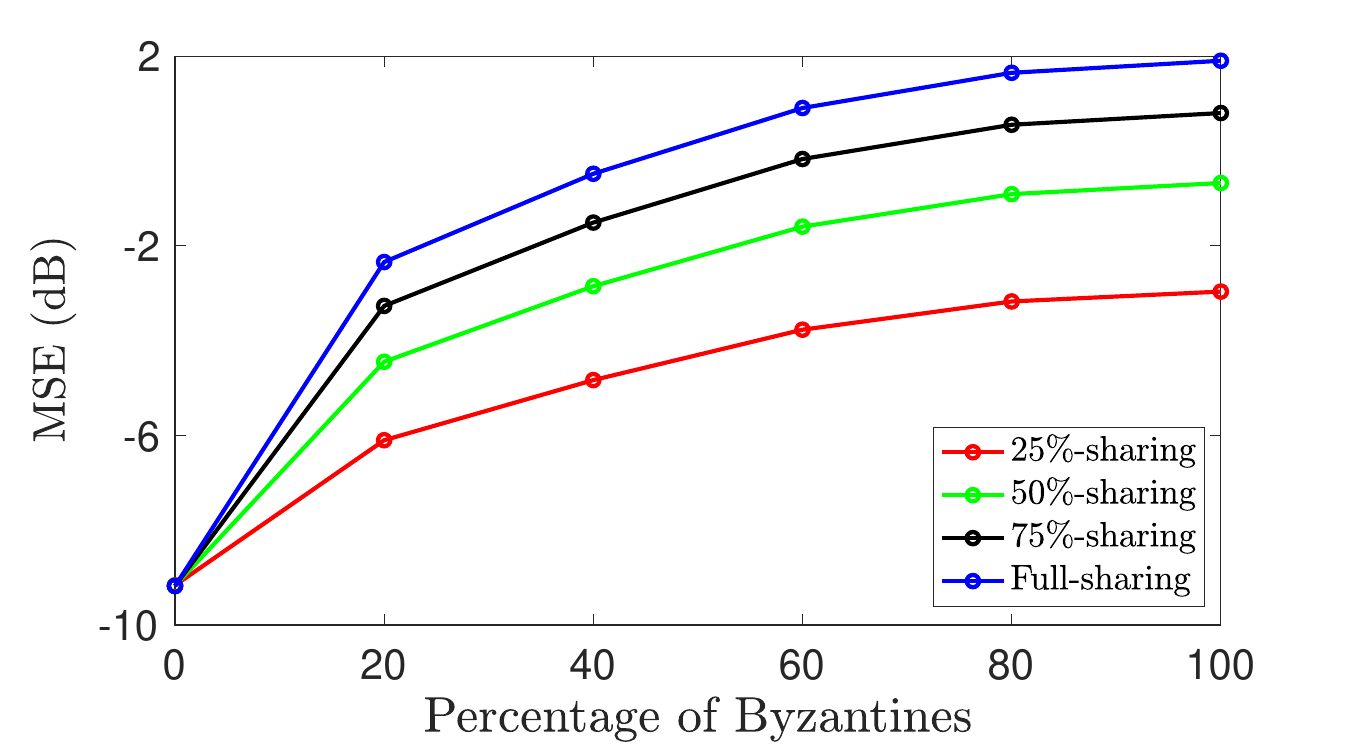}
	\caption{MSE versus percentage of the Byzantine agents in the network.}
	\label{Fig_8}
\end{figure}

For different values of $l$, Fig.~\ref{Fig_7} plots the $\text{MSE}^\prime$ versus time index $k$ for optimized and random selection of the attack covariance.  It can be seen that when less information is shared, the sensitivity to perturbation sequences with optimized covariance increases, resulting in high levels of fluctuation in the $\text{MSE}^\prime$. In addition, Figs.~\ref{Fig_5}~and~\ref{Fig_6} show that the optimized selection matrices have a greater impact when less information is shared, e.g., $25\%$ and $50\%$-sharing, while optimal attack covariance has a higher impact when larger fractions of information are shared, e.g., $75\%$ and full-sharing.

\begin{figure}[!t]
	\centering
	\includegraphics[width=.5\textwidth]{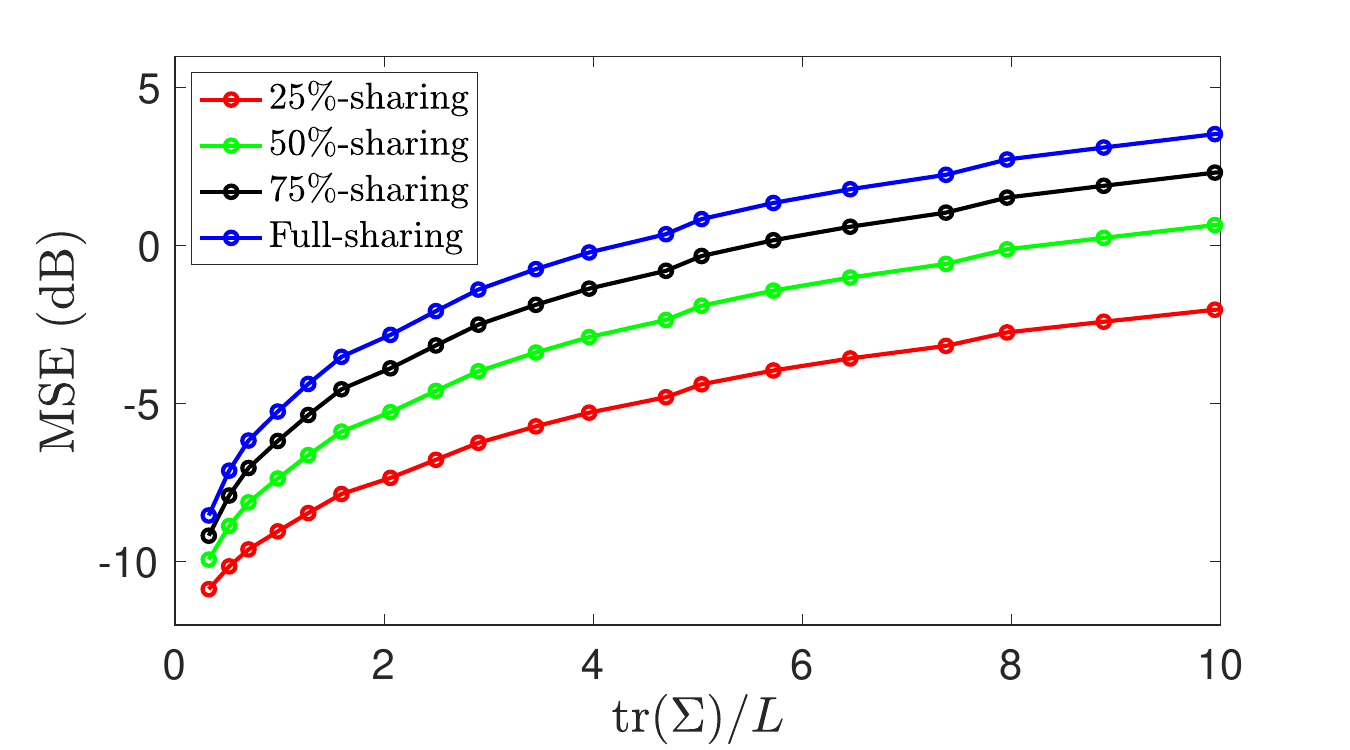}
	\caption{MSE versus  trace of the attack covariance, i.e., $\text{tr}(\mathbf{\Sigma})/L$.}
	\label{Fig_9}
\end{figure}

In order to analyze the robustness of the proposed BR-CDF algorithm to the number of Byzantine agents, Fig.~\ref{Fig_8}  plots the  MSE versus the percentage of Byzantine agents in the network. As expected, we see that as the percentage of Byzantine agents increases, the MSE grows; however,  partial sharing of information can significantly improve the resilience to Byzantine attacks, as illustrated by obtaining the lower MSE. In addition, Fig.~\ref{Fig_9} illustrates the MSE  versus the trace of the attack covariance in order to assess the robustness of the  BR-CDF algorithm to perturbation sequences. It can be seen that partial sharing of information improves robustness to injected noise by obtaining lower MSE.

\section{Conclusion}
This paper proposed a Byzantine-resilient consensus-based distributed filter (BR-CDF) that allows agents to exchange a fraction of their information at each time instant. We characterized the performance and convergence of the BR-CDF and investigated the impact of coordinated data falsification attacks. We showed that partial sharing of information provides robustness against Byzantine attacks and also reduces the communication load among agents by sharing a smaller fraction of the states at each time instant. Furthermore, we analyzed the worst-case scenario of a data falsification attack where Byzantine agents cooperate on designing the covariance of their falsification data or the sequence of their shared fractions. Finally, the numerical results verified the robustness of the proposed BR-CDF against Byzantine attacks and corroborated the theoretical findings.

\bibliographystyle{IEEEtran}
\bibliography{Ref.bib}

\end{document}